\documentclass[11pt]{article}

\usepackage[a4paper,margin=2.5cm]{geometry}
\usepackage[T1]{fontenc}
\usepackage{lmodern}
\usepackage{microtype}
\usepackage{amsmath,amssymb,amsthm}
\usepackage{booktabs}
\usepackage{xspace}
\usepackage{pgfplots}
\usepackage[hidelinks]{hyperref}
\pgfplotsset{compat=1.18}
\usetikzlibrary{arrows.meta,decorations.pathreplacing}

\def\..{\,\mathpunct{\ldotp\ldotp}} 

\newcommand{\PHast}{PHast\xspace}
\newcommand{\PHastP}{PHast$^+$\xspace}
\newcommand{\PHastR}{PHastR\xspace}
\newcommand{\hi}{\operatorname{hi}}
\DeclareMathSymbol{\ll}{\mathbin}{symbols}{"1C}
\DeclareMathSymbol{\gg}{\mathbin}{symbols}{"1D}
\newtheorem{proposition}{Proposition}

\title{\PHastR: Faster Minimal Perfect Hashing\\with Rotated Layouts\thanks{Previous
versions of this paper reported experiments on an Amazon EC2
\texttt{c7i.metal-24xl} instance (Intel Xeon Platinum 8488C). This version
reports experiments on an \texttt{m8a.metal-24xl} instance (AMD EPYC 9R45),
with an implementation that removes some bottlenecks revealed by that
processor.}}
\author{Sebastiano Vigna\\Dipartimento di Informatica, Università degli Studi di Milano}
\date{}

\begin{document}
\bibliographystyle{plain}
\maketitle

\begin{abstract}
We present \PHastR, a minimal perfect hash function of the \PHast family. Like
\PHastP, the fast-to-build variant of \PHast, \PHastR finds feasible seeds with
bit-parallel operations: its seeds move keys along \emph{rings} of slots, and
to evaluate the seeds of a bucket it cyclically shifts, for each key, a word
recording which slots of its ring are used. Unlike \PHastP, however, it is
almost immune to \emph{self-collisions}, which make two keys of a bucket
collide for (almost) all its seeds: more than a third of the buckets bumped by \PHastP
with wrapping and multiplier~$3$ contain one. With 8-bit seeds, \PHastR uses
$1.94$ bits per key, essentially the same space as the best variant of
\PHastP with wrapping, but it is built in
$37\%$ of the time, and its queries are $34\%$ faster on $10^9$ keys, as it
bumps $0.5\%$ of the keys instead of $3.9\%$; they are faster also than those
of \PHast, which takes almost thirteen times as long to build, and than those of
PtrHash, which uses more than one bit per key more. With a larger
expected bucket size
\PHastR comes within $0.02$ bits per key of the space of \PHast, using
$1.90$ bits per key, and its queries are still about as fast. Its
construction also scales better with the number of threads: with all the
threads of a 96-core machine it is more than two and a half times faster than that of
\PHastP. An offline construction, which keeps the hashes of the keys on
disk and needs less than half a byte per key of memory, built \PHastR on
a trillion keys in $9.3$ hours on a single machine. Finally, we argue that
when each key can be placed only within a range of $W$ slots, every
placement of $n$ random keys on $n$ slots, however found, must bump on
large key sets a constant fraction of the keys: we derive an asymptotic
lower bound of about $0.1/W$ on the expected value of this fraction, and we
prove that, in a Poisson model, the best placement bumps about $1/2W$ of the
keys, which is about one tenth of the keys bumped by \PHastR.
\end{abstract}

\section{Introduction}

A \emph{minimal perfect hash function} (MPHF) maps a static set of $n$ keys
bijectively to $[0\..n)$; it requires at least $\lg e\approx 1.443$
bits per key~\cite{FKSSSFPHF} (we use $\lg$ for the binary logarithm and $\ln$ for the natural
logarithm). Among the many practical constructions, \emph{bucket placement}
gives the fastest queries: keys are hashed to buckets, and a per-bucket seed
chooses for each key of the bucket a \emph{slot} of $[0\..n)$, so that no two
keys collide.
The idea goes back to FCH~\cite{FCHFACMPHF} and to Pagh's
hash-and-displace~\cite{PagHDEEMPHF}, and it has been refined by many
constructions since (Section~\ref{sec:related}). The main goal of these
constructions is not to approach the lower bound, as those based on recursive
splitting do at the price of slower queries, but to obtain a practical
structure that answers a query with close to one memory access on average.

\PHast~\cite{BeSPHast,BeSPHastALENEX} pushes this approach to its limit: the seeds are
stored in a plain array of fixed-width values, each seed places the keys of
its bucket inside small overlapping slices of the output range, and a query
reads a single seed for almost all keys. Finding the seeds, however, is
expensive, as \PHast tries all the seeds of each bucket: construction takes
hundreds of nanoseconds per key. Its variant \PHastP finds feasible seeds
with bit-parallel operations, and it is an order of magnitude faster to
build, but with 8-bit seeds it needs $2.115$ bits per key instead of $1.885$.
Its variants \emph{with wrapping}, in which keys wrap around the end of their
slices, need less space ($1.939$ bits per key for the best one, that is, the
one using the least space), but they are slower to build.

In this paper we present \PHastR, a structure of the same family that, with
8-bit seeds, has almost the space of \PHast, faster queries, and a
construction faster than that of the variants of \PHastP with wrapping and,
with several threads, than that of \PHastP itself.

The idea behind \PHastR lies between \PHast and \PHastP. In \PHastP a seed
moves all the keys of a bucket by the same amount: this makes it possible to
evaluate many seeds at once with bit-parallel operations, but two keys of a
bucket that collide for one seed collide for (almost) all seeds. In \PHast,
instead, each seed places the keys independently, but seeds must be tried
one at a time. In \PHastR the lowest bits of a seed choose one of a few
independent \emph{layouts}, each assigning an offset to every key, and the
remaining bits choose a \emph{rotation}, which moves all the keys of the
bucket by the same amount along their \emph{rings} of slots: all the rotations of a layout can be evaluated at
once, as in \PHastP, but two keys collide for (almost) all seeds only if they
do so in every layout, which is very unlikely. Thus, \PHastR bumps far fewer
keys than \PHastP, and since, once the seed width and the expected bucket
size are fixed, bumping is where both the excess space and the slow queries
come from, it gets almost the space of \PHast with faster queries; choosing
a layout requires just two shifts or, with the default parameters, a 16-bit
load.

This paper makes four contributions:
\begin{itemize}
\item \PHastR and its bit-parallel search, which finds all the rotations of a
layout that map no key to a used slot by cyclically shifting, for each key,
the bits recording which slots of its ring are used
(Section~\ref{sec:phastr});
\item a comparison with \PHast and \PHastP, showing that a large part of the
bumping of \PHastP is caused by self-collisions inherent to its additive
scheme, and that layouts make their effect negligible
(Section~\ref{sec:differences});
\item an implementation in the Rust crate
\texttt{sux}\footnote{\url{https://crates.io/crates/sux}}, an experimental
comparison with the reference implementation of \PHast and \PHastP and with
PtrHash (Section~\ref{sec:experiments}), and an offline construction that reads the
keys just once, needs on large key sets less than half a byte per key of
memory, and builds the same
structure as the construction in memory (Section~\ref{sec:offline});
\item an analysis of the number of keys that every placement must bump when
each key can be placed only within a range of $W$ slots: on large key sets
we derive an asymptotic lower bound of about $0.1/W$ on the expected fraction
of bumped keys, and we prove that, in a Poisson model, the best placement
bumps about $1/2W$ of the keys, so no
search can produce a first level free of bumping, but \PHastR bumps about ten times as many keys as this constraint
forces (Section~\ref{sec:limits}).
\end{itemize}

\section{Related Work}
\label{sec:related}

We refer the reader to the survey by Lehmann et
al.~\cite{LMPPSVWMMPHS} for a complete account of the field, and recall here
the approaches that are relevant to this paper. MPHFs based on
\emph{fingerprinting}, such as BBHash~\cite{LRCPFSMPHMKS} and
FMPH~\cite{BelFBMPHR}, map the keys to bit arrays of decreasing size, one
level at a time: keys that collide in one array are \emph{bumped} to the next array;
they are fast to build, but they use more than two bits per
key. At the other end, MPHFs based on \emph{recursive splitting}, such as
RecSplit~\cite{EGVRMPHRS} and its parallel implementation~\cite{BKLSHPCRBMPHF},
ShockHash~\cite{LSWSTOSMPHBBF}, and CONSENSUS-RecSplit~\cite{LSWCSESAMPH}, get
close to the lower bound of $\lg e$ bits per key, at the price of slower
queries; SicHash~\cite{LSWSSICTPH} combines small irregular cuckoo hash
tables with retrieval data structures.

\emph{Bucket placement} goes back to FCH~\cite{FCHFACMPHF}, which, however,
was slow to build unless its functions used a lot of space, and to Pagh's
hash-and-displace~\cite{PagHDEEMPHF}, in which each bucket stores a displacement
that moves all its keys by the same amount: two keys of a bucket that collide
for one displacement thus collide for all of them, and the scheme requires the
pair of hash functions locating the keys to be injective on the key set.
CHD~\cite{BBDHDC} brought the idea back: the seed of a bucket selects a function
from a family of random functions (in practice, linear combinations of two hash
functions of the key with coefficients depending on the seed), so that, in the
analysis, different seeds place the keys of a bucket independently; seeds are tried one at a time,
in order, and then compressed. PTHash~\cite{PTPRFMPH} revisited FCH with a
careful engineering of the encoding of the seeds (its \emph{pilots}) and of the
distribution of the keys into buckets, revitalizing the field; it was followed by
PHOBIC~\cite{HLPPPHOBSIC}, which optimizes the bucket sizes and interleaves the
encoding of the pilots, and by PtrHash~\cite{RagPTR}, which stores the pilots in
bytes and replaces the greedy placement of the buckets with a cuckoo-like search
that evicts already placed buckets, obtaining very fast queries. The
extension of PtrHash to non-minimal $k$-perfect hashing~\cite{GHSWNMKPHTLBAFSHT}
uses a variant of hash-and-displace: the highest bit of the seed of a bucket
chooses between two positions of each key, computed by a further hash
function, and the remaining bits move all the keys of the bucket by the same
number of bins, so that, as in \PHastP, the displacements can be filtered
with bitmasks; since keys are not confined to slices, the displacements do not
wrap around. Finally, \PHast~\cite{BeSPHast,BeSPHastALENEX} combines fixed-width seeds and slices with bumping: buckets
that cannot be placed are left to the following levels, so that the search never
backtracks.

\PHastR inherits the structure of \PHast, and its seeds lie between two
extremes: the additive displacements of hash-and-displace and \PHastP, and
the independent functions of \PHast and of the analysis of CHD. Like the
linear combinations used in practice by CHD, a seed
both changes the relative positions of the keys of a bucket (by choosing a
layout) and moves them all by the same amount (by choosing a rotation), but
the few layouts are obtained for free from the computation of the bucket,
and all the rotations of a layout are evaluated at once with bit-parallel
operations.

\section{\PHastR}
\label{sec:phastr}

\PHastR follows the scheme of \PHast~\cite{BeSPHast,BeSPHastALENEX}; we discuss the
differences in Section~\ref{sec:differences}. Each key has a $64$-bit hash,
which determines its \emph{bucket}, among $B$ buckets, and its \emph{slice},
an interval of $L$ consecutive slots of the \emph{output range} $[0\..m)$
whose start grows linearly with the hash. Each bucket has a seed, stored in a
plain array of $S$-bit values; the seed maps each key of the bucket to a slot
of its slice. A slot is \emph{used} if a key has already been mapped to it, and
\emph{free} otherwise; a seed is \emph{feasible} for a bucket if it maps its
keys to distinct free slots.

The structure is a sequence of \emph{levels}
using the same scheme: buckets for which no seed is feasible are
\emph{bumped} (seed~$0$), and their keys are handled by the following
levels. The first level has $m=n$, so every bumped key leaves a
free slot (a \emph{hole}) in the first level, and an Elias--Fano
sequence~\cite{EliESRCASF,FanNBRISM,VigBIRSQ}, the \emph{remapping}, maps the slots of
the following levels, in order, to the holes. Since bumped keys are few, a query reads a
single seed for almost all keys. Seeds are assigned by a \emph{sweep} over
the buckets: buckets enter, in order of index, the \emph{window}, a priority
queue that favors large buckets and holds the buckets not yet placed among
the $256$ starting from the first bucket not yet placed; the sweep
repeatedly extracts the bucket of highest priority and gives it the feasible
seed minimizing the product of the distances of the slots of its keys from
a point slightly before the slices of the bucket. To place a bucket the sweep
needs just the hashes of its keys: for the time being, the reader can assume
that the hashes are sorted by bucket (in any order within a bucket); later in
this section we explain how we sort them cheaply, spreading a bucket sort over
the sweep.

Each hole costs about $2+\lg(1/\beta)$ bits in the remapping,
where $\beta$ is the fraction of bumped keys, and the bumped keys must be
placed again by the following levels, which in turn bump a fraction $\beta$
of their keys, so that overall the levels receive about
$n(1+\beta+\beta^2+\cdots)=n/(1-\beta)$ keys. With $S$-bit seeds and expected
bucket size $\lambda=n/B$ the seeds take $S/\lambda$ bits for each key a level
receives,
and the space in bits per key is thus approximately
\begin{equation}
\label{eq:space}
\frac{S}{\lambda(1-\beta)} + \beta\Bigl(2+\lg\frac1\beta\Bigr),
\end{equation}
and since the entropy of a random set of $\beta n$ holes is already about
$\beta n(\lg(1/\beta)+\lg e)$, the encoding of holes cannot be
substantially improved: once $S$ and $\lambda$ are fixed, \emph{the only
lever is $\beta$}. We will show in Section~\ref{sec:limits} that on large key
sets no placement, however found, can avoid bumping altogether; the purpose of
the design of \PHastR is to bump few keys while finding feasible seeds with
bit-parallel operations.

\begin{figure}[t]
\centering
\begin{tikzpicture}[x=4mm,y=4mm,font=\footnotesize,>={Latex[length=1.3mm]}]
\node[left] at (-.2,8) {buckets};
\fill[gray!35] (9,7.6) rectangle ++(3,.8);
\foreach \i in {0,...,9} \draw (3*\i,7.6) rectangle ++(3,.8);
\foreach \i/\l in {2/b-1,3/b,4/b+1,5/b+2} \node[font=\scriptsize] at (3*\i+1.5,8) {$\l$};
\draw[decorate,decoration={brace,amplitude=4pt}] (9,8.6) -- (25,8.6)
  node[midway,above=4pt,font=\scriptsize] {buckets whose keys can reach slot $p$};
\node[left] at (-.2,6) {slots};
\fill[gray!35] (24,5.6) rectangle ++(1,.8);
\foreach \i in {0,...,31} \draw (\i,5.6) rectangle ++(1,.8);
\node[below,font=\scriptsize] at (24.5,5.6) {$p$};
\draw[dashed] (9,7.6) -- (9,6.4) (12,7.6) -- (12,6.4);
\draw[->] (4.5,4.6) node[below,font=\scriptsize] {origin of $b$} -- (4.5,5.6);
\foreach \s/\y/\k in {9/3.6/0,10/2.8/1,11/2.0/2} {
  \fill[gray!12] (\s,\y) rectangle ++(16,.5);
  \draw (\s,\y) rectangle ++(16,.5);
  \node[left,font=\scriptsize] at (8.9,\y+.25) {$x_\k$};
}
\foreach \p in {11,15,19,23} \fill (\p+.5,3.85) circle (1.1mm);
\draw (15.5,3.85) circle (1.9mm);
\draw[<->] (15.5,4.6) -- node[above,font=\scriptsize] {$T$} (19.5,4.6);
\end{tikzpicture}
\caption{\label{fig:geometry}Buckets, slices and strides in a level with
output range $[0\..m)$, $m\approx n$, expected bucket size $\lambda=3$,
slices of length $L=16$, and stride $T=4$. The top line shows some
consecutive buckets, the middle line the slots of the output range, and the
bars at the bottom the slices of the keys $x_0$, $x_1$, $x_2$ of
bucket~$b$ (gray). As both the bucket of a key and the start of its slice
grow linearly with its hash, each bucket corresponds to an interval of about
$\lambda$ slice starts (the dashed lines delimit that of bucket~$b$), and
the slices of the keys of bucket~$b$ start a few slots apart. In a layout,
a key moves along its ring, the slots of its slice whose distance from its
base slot (circled) is a multiple of~$T$ (the dots, for $x_0$). The seed of
bucket~$b$ is chosen using the distances of the slots of its keys from its
origin, which precedes its leftmost slice ($95$ slots before in our
implementation). Only the keys of the about $L/\lambda$ buckets whose
slices can contain a slot~$p$ (gray) can be placed in it.}
\end{figure}
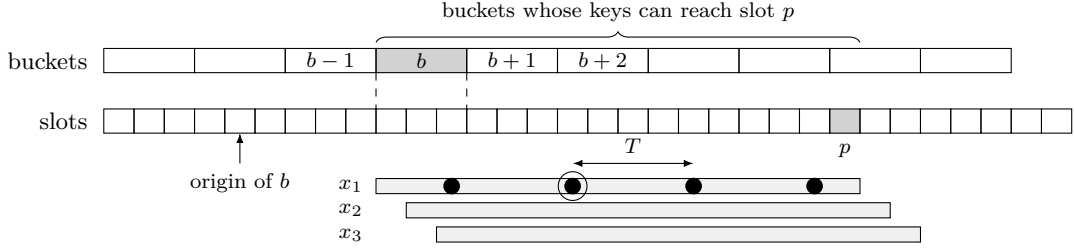

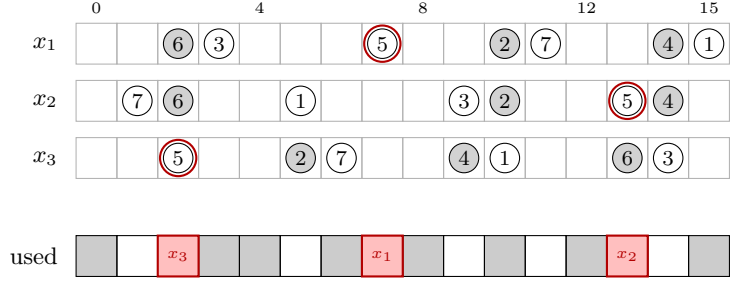
\begin{figure}[t]
\centering
\begin{tikzpicture}[x=5.4mm,y=5.4mm,font=\footnotesize,
  p0/.style={circle,draw,fill=gray!35,inner sep=0pt,minimum size=3.9mm,font=\scriptsize},
  p1/.style={circle,draw,fill=white,inner sep=0pt,minimum size=3.9mm,font=\scriptsize}]
\foreach \y/\k in {5.2/0,3.8/1,2.4/2} {
  \node[left] at (-.2,\y+.5) {$x_\k$};
  \foreach \i in {0,...,15} \draw[gray!60] (\i,\y) rectangle ++(1,1);
}
\foreach \i in {0,4,8,12,15} \node[above,font=\tiny] at (\i+.5,6.2) {$\i$};
\foreach \p/\s in {10/2,14/4,2/6} \node[p0] at (\p+.5,5.7) {$\s$};
\foreach \p/\s in {15/1,3/3,7/5,11/7} \node[p1] at (\p+.5,5.7) {$\s$};
\foreach \p/\s in {10/2,14/4,2/6} \node[p0] at (\p+.5,4.3) {$\s$};
\foreach \p/\s in {5/1,9/3,13/5,1/7} \node[p1] at (\p+.5,4.3) {$\s$};
\foreach \p/\s in {5/2,9/4,13/6} \node[p0] at (\p+.5,2.9) {$\s$};
\foreach \p/\s in {10/1,14/3,2/5,6/7} \node[p1] at (\p+.5,2.9) {$\s$};
\node[left] at (-.2,.5) {used};
\foreach \i in {0,3,4,6,8,10,12,15} \fill[gray!40] (\i,0) rectangle ++(1,1);
\foreach \i in {0,...,15} \draw (\i,0) rectangle ++(1,1);
\foreach \p/\k in {7/0,13/1,2/2}
  \draw[red!70!black,thick,fill=red!25] (\p,0) rectangle ++(1,1)
    node[midway,font=\tiny,text=red!70!black] {$x_\k$};
\foreach \p/\y in {7/5.7,13/4.3,2/2.9}
  \draw[red!70!black,thick] (\p+.5,\y) circle (2.4mm);
\end{tikzpicture}
\caption{\label{fig:placement}The slots of the three keys $x_0$, $x_1$, $x_2$
of a bucket for each seed, with $3$-bit seeds, $R=2$ layouts, and slices of
length $L=16$ (for simplicity, the three slices coincide). Each of the top
three lines shows the slice of a key, with the positions of some slots
above, and a circle containing a seed~$s$ marks the slot to which $s$ maps
the key. Even seeds (gray circles) choose layout $0$, odd seeds (white
circles) layout $1$: seed $s=2j+r$ chooses layout $r$ and rotation $j$,
which moves the keys by $jT$ slots from their base slots, where $T=4$ is the
stride, so in each layout a key goes through the four slots of its ring
(only three of them in layout $0$, as seed $0$ is not available). In layout
$0$ the keys $x_0$ and $x_1$ have the same base slot, so they collide for
all even seeds; in layout~$1$ they do not. The bottom line shows the slots
that are already used (gray): the only feasible seed is $5$, which maps the
keys to the free slots marked in red.}
\end{figure}

\paragraph{Placement scheme.}
We denote with $\hi(x)$ and $\operatorname{lo}(x)$ the upper and lower $64$
bits of a $128$-bit product. Consider a level with output range $[0\..m)$ and
$B$ buckets, whose slices have length $L=2^\ell$ (slice lengths are powers of
two; we discuss their choice below). Let $h$ be the hash of a key, and let
$\alpha=h/2^{64}\in[0\..1)$. The bucket of the key is $\lfloor\alpha
B\rfloor$, which in fixed-point arithmetic is just $\hi(hB)$, and the start of
its slice is $\sigma(h)=\lfloor\alpha(m-L+1)\rfloor$, which is
$\hi(h(m-L+1))$, so that slices lie within the output range
(Figure~\ref{fig:geometry}).\footnote{The assignment of buckets by a fixed-point
multiplication, which makes each bucket an interval of hashes, was introduced
by Genuzio, Ottaviano and Vigna for the offline bucketing of
GOV~\cite{GOVFSCCSMPHF}.} Moreover, there
are $R$ \emph{layouts}, where $R=2^\rho$ is a power of two, each assigning
an in-slice offset to every key:
letting $o=\operatorname{lo}(hB)$, we divide $o$ into $R$ fields of
$f=64/R$ bits, and layout $r\in[0\..R)$ assigns to the key the offset
\[
\omega_r(o)=(o\gg rf) \bmod L,
\]
that is, the $\ell$ lowest bits of field $r$. We require $\ell\le f$, so
that the offsets of different layouts use disjoint bits of $o$.

The seed $s\in\bigl[1\..2^S\bigr)$ of a bucket has two parts: its lowest $\rho$ bits,
$r=s\bmod R$, choose a layout, and its remaining bits,
$j=s\gg\rho\in\bigl[0\..2^{S-\rho}\bigr)$, a \emph{rotation}. Assuming for the
time being that $\ell\geq S$, let $u=2^{\ell-S}$: the seed maps each key of
the bucket to the slot
\[
\sigma(h)+\bigl(\omega_r(o)+su\bigr)\bmod L
=\sigma(h)+\bigl(\omega_r(o)+ru+jT\bigr)\bmod L,
\]
where $T=Ru$ is the \emph{stride}; the two forms are equal because $s=jR+r$,
and queries compute the first one.

In words: once the layout is chosen, each
key of the bucket has a \emph{base slot}, the slot for rotation $0$ (as
offsets are uniform, the constant $ru$ is immaterial), and rotation~$j$ moves
all the keys of the bucket by $jT$ slots, wrapping around the end of their
slices. As the rotation varies, a key goes
exactly once through the $2^{S-\rho}=L/T$ slots of its slice whose distance
from its base slot is a multiple of $T$: we call these slots its \emph{ring}
in layout~$r$ (Figure~\ref{fig:placement}). Seed $0$, that is, layout $0$
with rotation $0$, marks bumped buckets, so layout $0$ has one rotation less.

Two keys of a bucket with the same base slot in a layout move together,
so they collide for all rotations but the few for which exactly one of them
wraps around (the slices of the keys of a bucket can start a few slots
apart): this is a \emph{self-collision}. Since the offsets of different
layouts use disjoint bits of $o$, self-collisions happen
independently in each layout: a bucket of size $k$ has a self-collision in
every layout with probability about $(k(k-1)/2L)^R$, which is negligible
already for $R=2$. Moreover, the seeds of a layout move all the keys of the
bucket by the same amount, so they are strongly correlated, whereas seeds of
different layouts are (almost) independent.

The value $o$ costs nothing, as it is the other half of the product that
yields the bucket, and it is a good source of offsets: for the keys of bucket
$b$ we have $hB=b\cdot2^{64}+o$, so $o/2^{64}=\alpha B-b$ is the fractional
part of $\alpha B$, which is uniformly distributed if $h$ is. Moreover, the
offsets of the keys of a bucket are independent of each other for every
layout, which would not be true of offsets taken from the upper bits of $h$,
as the keys of a bucket share the upper ${\approx}\lg B$ bits of
$h$.

We require the number $B$ of buckets to be odd: if $B$ is even, the
lowest bit of $o$, and thus of the offset of layout~$0$, is always zero,
whereas if $B$ is odd the lowest bits of $o$ are a bijective function of the
lowest bits of $h$.

\paragraph{Queries.}
A query computes $hB$, reads the seed $s$ of bucket $\hi(hB)$, and, unless
$s=0$, computes the first form of the formula above. There is no need to
extract $r$: since $sf=(jR+r)f=64j+rf$, we have $sf\bmod 64=rf$, and the
query computes
\[
\sigma(h)+\bigl((o\gg(sf\bmod 64))+su\bigr)\bmod L,
\]
where a single reduction modulo $L$ both extracts the offset $\omega_r(o)$
from the shifted $o$ and wraps the sum around the end of the slice. The
count of this shift, however, depends on the seed, and thus on a cache miss,
which reduces memory-level parallelism. Thus, in the first level, if $R=4$
and $u=4$, as with $S=8$ and $L=1024$, on little-endian machines the query
stores $o$ in memory and reads field $r$ with a $16$-bit load at index
$s\bmod 4$ (satisfied by store-to-load forwarding), of which the reduction
modulo $L$ keeps the lowest $\ell=10$ bits: on $10^9$ keys, this makes
queries about $10\%$ faster.

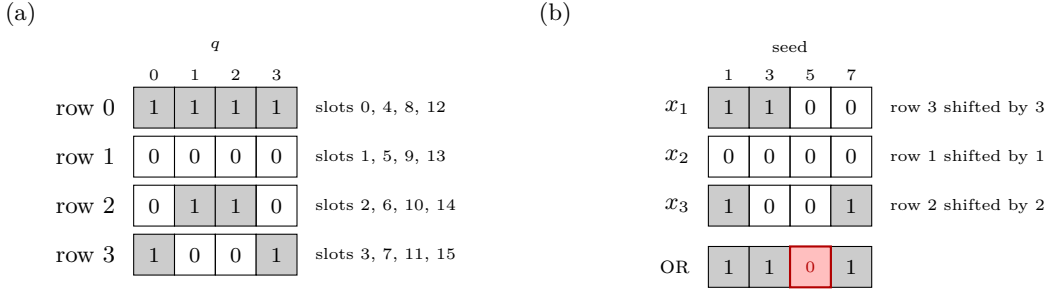
\begin{figure}[t]
\centering
\begin{tikzpicture}[x=5.4mm,y=5.4mm,font=\footnotesize,>={Latex[length=1.6mm]}]
\def\bits#1#2#3#4{
  \node[left] at (#2-.2,#3+.5) {#1};
  \foreach \v [count=\d from 0] in {#4} {
    \ifnum\v=1 \fill[gray!40] (#2+\d,#3) rectangle ++(1,1); \fi
    \draw (#2+\d,#3) rectangle ++(1,1);
    \node[font=\scriptsize] at (#2+\d+.5,#3+.5) {$\v$};
  }}
\node[anchor=west] at (-3.4,6.4) {(a)};
\node[font=\tiny] at (2,5.6) {$q$};
\foreach \d in {0,...,3} \node[font=\tiny] at (\d+.5,4.9) {$\d$};
\bits{row $0$}{0}{3.6}{1,1,1,1}
\bits{row $1$}{0}{2.4}{0,0,0,0}
\bits{row $2$}{0}{1.2}{0,1,1,0}
\bits{row $3$}{0}{0}{1,0,0,1}
\foreach \y/\l in {3.6/{0, 4, 8, 12},2.4/{1, 5, 9, 13},1.2/{2, 6, 10, 14},0/{3, 7, 11, 15}}
  \node[right,font=\tiny] at (4.2,\y+.5) {slots \l};
\begin{scope}[xshift=76mm]
\node[anchor=west] at (-4.4,6.4) {(b)};
\node[font=\tiny] at (2,5.6) {seed};
\foreach \s [count=\d from 0] in {1,3,5,7} \node[font=\tiny] at (\d+.5,4.9) {$\s$};
\bits{$x_0$}{0}{3.6}{1,1,0,0}
\bits{$x_1$}{0}{2.4}{0,0,0,0}
\bits{$x_2$}{0}{1.2}{1,0,0,1}
\bits{\textsc{or}}{0}{-.3}{1,1,0,1}
\draw[red!70!black,thick,fill=red!25] (2,-.3) rectangle ++(1,1)
  node[midway,font=\tiny,text=red!70!black] {$0$};
\foreach \y/\l in {3.6/{row $3$ shifted by $3$},2.4/{row $1$ shifted by $1$},1.2/{row $2$ shifted by $2$}}
  \node[right,font=\tiny] at (4.2,\y+.5) {\l};
\end{scope}
\end{tikzpicture}
\caption{\label{fig:search}Finding the free rotations of layout $1$ for the
bucket of Figure~\ref{fig:placement}. (a) The set of used slots of
Figure~\ref{fig:placement} is stored by residue classes modulo the stride
$T=4$: bit $q$ of row $\tau$ records whether slot $qT+\tau$ is used (the
labels on the right list the slots of each row). (b) The bits of the ring
of a key form a row of (a) (in general, $L/T$ consecutive bits of a row),
and the rotations go through them cyclically, starting from the ring
position of the key: for example, the base slot of $x_0$ is slot $15$
(seed $1$), which is slot $3$ of its ring (row $3$), so the ring position of
$x_0$ is $3$. Each line shows the bits of the ring of a key cyclically
shifted by its ring position (the labels on the right give the row and the
shift), so that column $j$ corresponds to rotation $j$, that is, to seed
$2j+1$. These bits are combined by a bitwise \textsc{or}: the zeros of the
result are the free rotations (here, only rotation $2$, that is, seed $5$,
in red).}
\end{figure}

\paragraph{Search.}
In a bit vector marking the used slots, the bits of the slots of a ring would be
$T$ bits apart. We make them adjacent by storing the set of used slots by
residue classes modulo $T$, as suggested by Beling and Sanders for \PHastP
with wrapping~\cite{BeSPHast,BeSPHastALENEX}, that is, as $T$ \emph{rows} in which bit $q$ of
row $\tau$ records whether slot $qT+\tau$ is used. All the slots of a
ring have the same residue modulo $T$ (wrapping around the end of a slice
does not change it, as $T$ divides $L$), so the \emph{bits of the ring} of a
key, which record whether the slots of its ring are used, are $L/T$
consecutive bits of a row. We number the slots of a ring from $0$ in order
of position in the slice, and we call \emph{ring position} of a key the
number $y$ of its base slot: rotation $j$ maps the key to slot
$(y+j)\bmod(L/T)$ of its ring. Thus, cyclically shifting the bits of the
ring so that bit $y$ becomes bit $0$,
we obtain $L/T$ bits whose bit $j$ tells whether rotation $j$ maps the key to
a used slot, and combining them for all the keys of the bucket by a bitwise
\textsc{or} we obtain $L/T$ bits whose zeros are the \emph{free} rotations,
which map every key to a free slot (Figure~\ref{fig:search}). A free rotation
$j$ yields the feasible seed $jR+r$ unless two keys of the bucket are mapped
to the same slot, which is rare. For $S=8$ and $R=4$ a ring has
$2^{S-\rho}=64$ slots whatever the slice length, except on small levels
(see below), so its bits fill a word:
each layout requires just reading and cyclically shifting a word for each
key.

Among the free rotations of all layouts we choose, as does the current version
of the reference implementation of \PHast~\cite{BelBRLPFSDS}, the one
minimizing the product of the distances of the slots of the keys from the
\emph{origin} of the bucket, which lies $95$ slots before the start of the
leftmost slice that a key of the bucket can have
(Figure~\ref{fig:geometry}). We
compute the logarithm of the product, that is, the sum of the base-2
logarithms of the distances, which we read, in fixed point with $24$
fractional bits, from a table: since it is a sum of integers, the cost of a
rotation does not depend on the order of the keys, which will be essential
for the offline construction of Section~\ref{sec:offline}. 

Evaluating a free
rotation thus requires a table lookup for each key of the bucket. Ties are
broken in favor of the first layout and, within a layout, of the first
rotation. When a bucket has at
most $32$ free rotations, which is the common case, we evaluate all of them,
with a single loop over the free rotations of all layouts that contains no
other branch. When there are more free rotations (e.g., at the beginning of a
sweep) we use instead the fact that in a layout the cost increases with the
rotation, except at the rotations at which some key wraps around: thus, for a
bucket of size $k$ we need to consider only the first free rotation and, for
each key, the first free rotation not smaller than the one at which the key
wraps around, that is, at most $k+1$ free rotations per layout. Finally, when
marking the slots of the best free rotation as used we check that no two keys
of the bucket are mapped to the same slot, and otherwise we move to the next
best one.

\paragraph{Sorting by bucket and scheduling.}
The sweep needs the hashes sorted by bucket, in any order within a bucket,
and we do not want to keep more than one copy of them. Since hashes are
uniform and buckets are intervals of hashes, we use a bucket sort, which we
spread over the sweep. We group the buckets into at most $256$ \emph{parts}, each made of a power-of-two
number of consecutive buckets, and we split the input, that is, the sequence
of keys, into \emph{batches}, one for each thread. There is a \emph{region}
for each pair of a batch and a part: a thread computes the hashes of the keys
of its batch, and appends each hash to the region associated with the batch
and with the part containing the bucket of the hash, recording the part in a
byte. The hashes in a region are thus in the order of the keys. Regions are
sized before computing the hashes, assuming that hashes are uniform, with a
slack of six standard deviations: the few hashes that do not fit their region,
if any, are kept separately.

A part is distributed into buckets only when a sweep needs it, into a buffer
of the sweep (the part is distributed first into smaller parts that fit the
cache, and then each smaller part into its buckets), so regions are never
modified. After the sweep, the keys of bumped buckets, which must be passed to
the next level, are found by walking again the keys of each batch in order:
the byte of a key gives its part, its hash is the next unread hash of the
region of the batch for that part, and a bit vector with a bit per bucket
tells whether the bucket of the hash has been bumped. Since a batch reads up
to $256$ regions at the same time, too many sequential streams for the
hardware prefetchers, the walk prefetches each region two cache lines ahead.
Writing to at most $256$
regions at a time and working on smaller parts of about $150$\,KB is
efficient on all the machines we tested, and the sort by bucket needs, besides the
hashes, little more than a byte per key: in our experiments, the whole
construction needs about $11$ bytes per key, excluding the keys
(Table~\ref{tab:threads}).

The priority of a bucket has the same form as in \PHast: a weight depending
on the size of the bucket, minus $1024$ times the index of the bucket. Since buckets enter the window in
order of index, the buckets of each size in it are in order of index, and the
bucket with the highest priority is the first bucket of some size. Thus,
instead of a priority queue we keep a first-in first-out queue for each size
(buckets with $64$ keys or more, which are extremely rare, are kept in a
heap): extracting the next bucket
requires scanning the first elements of a few queues, and adding a bucket
takes constant time.

The set of used slots of the level, which is needed to build the remapping,
is a byproduct of the sweep: the sweep keeps in a cyclic
buffer only the bits of the slots that the buckets still to be placed can
reach, and copies the bits of the slots that are no longer reachable to the
bit vector of the level when their space in the buffer is recycled.

\paragraph{Slice length.}
The slice length of a level with output range $[0\..m)$ is the smallest power
of two larger than $m/2$, up to a maximum, which is a parameter: we use
$L=1024$ for $S=8$ and $L=2048$ for $S=10$. The
slices of small levels can thus have fewer slots than there are seeds, that
is, $\ell<S$: in this case we let $u=1$, so the stride is $T=R$, rings have
$L/R$ slots (a single slot if $L<R$), and rotations $j$ and $j'$ map the
keys of a bucket to the same slots whenever $L$ divides $(j-j')R$.

\paragraph{Levels.}
The following levels use the same placement scheme with output range equal
to their number of keys, and the remapping maps their slots to the holes of
the first level. A bumped key is hashed again just once, with a different
hash seed: if $h'$ is this \emph{second hash}, the hash of the key at level
$i>1$, which plays the role of $h$ at that level, is
\[
h_i=\hi(Y)\oplus\operatorname{lo}(Y),\quad\text{where $Y=(h\oplus z_i)(h'\oplus c)$},
\]
$z_i$ is a \emph{salt} associated with the level, and $c$ is a constant (we
let $h_1=h$). Thus, the following levels depend only on $h$ and $h'$: once $h'$ has
been computed the key is no longer needed, and the code accessing the
following levels, which takes just two integers, can be kept separate from
that accessing the first level.

\paragraph{Colliding hashes.}
Keys with the same $64$-bit hash at the first level have the same slice and
the same offsets in all layouts, so they collide for every seed and their
bucket is bumped; at the following levels their hashes depend also on the
second hash, and thus they are separated. Since about $n^2/2^{65}$ pairs of
keys have the same hash (about three for $n=10^{10}$), the keys bumped in this
way are negligible with respect to the $\beta n$ keys that are bumped anyway.
Hence, $64$-bit hashes are sufficient for any practical number of keys.
Construction would fail only if two keys had the same pair of hashes, which
happens in practice only for duplicate keys. 

Since duplicate keys are always
bumped from the first level, to report them we examine the buckets of the
second level: as two distinct keys might have the same hash at that level,
keys of a bucket with the same hash are hashed again with a third hash seed,
and only keys whose hashes coincide again are considered duplicates.

\paragraph{Last level.}
When at most $4096$ keys are left, a last level with $n'$ keys uses the same
placement scheme with an output range of $\lfloor 5n'/4\rfloor+16$ slots, expected
bucket size $\min(\lambda,3)$, and
no bumping, changing the salt (and, every eight
attempts, enlarging the output range) in the unlikely case of failure.

\paragraph{Parallel construction.}
As in \PHast, buckets are split into chunks (by default, one for each
thread, provided that each chunk has at least $\max(4096,64g)$ buckets, where
$g$ is the length of a gap) separated by gaps of
\[
\left\lfloor \frac{LB}{m-L+1}\right\rfloor+1
\]
buckets: the slices of the keys of bucket $b$ start at about $b(m-L+1)/B$, and the slot
of a key lies at most $L-1$ slots after the start of its slice, so
(considering the rounding of slice starts) the keys of different chunks
cannot reach the same slots (in Figure~\ref{fig:geometry}, the brace spans
about as many buckets as a gap). Chunks are swept in parallel, and then all gaps
are swept in parallel, each with a new sweep in which the slots used by the
keys of the neighboring chunks (on both sides) are marked as used.

The structure built depends on the chunks, and thus, by default, on the
number of threads. To build the same structure with any number of threads,
the number of chunks can be fixed, and the chunks are then processed by the
available threads. Since each additional chunk adds a gap, whose buckets must
avoid the slots used by the chunks on both sides, more chunks make
construction slightly slower and the structure slightly larger.

\section{Differences from \PHast and \PHastP}
\label{sec:differences}

In \PHast~\cite{BeSPHast,BeSPHastALENEX}, in its variant \PHastP, and in \PHastR the bucket
of a key with hash $h$ is $\hi(hB)$, and the seed $s$ of the bucket maps the
key to slot $\sigma(h)+p(s,h)$, where $\sigma(h)$ is the start of its slice:
the three structures differ mainly in their placement scheme, that is, in the
function $p$. In \PHast, $p$ is a pseudorandom function, so finding the best
seed requires trying all seeds. \PHastP uses instead the \emph{additive
scheme} $p(s,h)=(h\bmod L)+s-1$, so that feasible seeds can be found by
combining, for each key, a word of the bit vector of used slots, but it needs
more space than \PHast; \PHastP
\emph{with wrapping} uses $p(s,h)=(h+\delta s)\bmod L$ for a small multiplier
$\delta$, so that keys do not leave their slice, and needs less space than
\PHastP without wrapping, but it is slower to build.

In \PHastR,
$p(s,h)=\bigl(\omega_r(o)+su\bigr)\bmod L$, where $r=s\bmod R$ and
$o=\operatorname{lo}(hB)$. Within a layout the scheme of \PHastR is thus
additive, like that of \PHastP with wrapping, so the free rotations of a
layout can be found with bit-parallel operations. In an additive scheme,
however, two keys of a bucket that collide for one seed collide for (almost)
all seeds: in the first level of the reference implementation of \PHastP with
8-bit seeds, wrapping, $\delta=3$, and $\lambda=5$, on $10^7$ keys, $35\%$
of the bumped buckets contain
such a self-collision. In \PHastR, instead, self-collisions happen
independently in each layout, and buckets with a self-collision in every
layout are negligible.

The construction of \PHastR differs from that of the reference implementation
also in several engineering aspects (Section~\ref{sec:phastr}): hashes are
sorted by bucket with a bucket sort spread over the sweep, rather than before
it, the priority queue is replaced by
a queue for each bucket size, bumped keys are hashed again just once rather
than at each level, the last level uses the same scheme as the other levels
rather than the pseudorandom scheme of \PHast, and the gaps between the
chunks of a parallel construction are swept in parallel, rather than serially. Slice lengths are those of \PHastP
with wrapping, but the expected bucket size is smaller, and the weights of
the bucket priorities are tuned for \PHastR.

\section{Experiments}
\label{sec:experiments}

We ran all timing experiments, except that on $10^{12}$ keys of
Section~\ref{sec:offline}, on an Amazon EC2 \texttt{m8a.metal-24xl} instance
(one AMD EPYC 9R45 with 96 cores, 384\,MiB of L3 cache, 32\,MiB for each
group of eight cores, and 384\,GiB of RAM), compiling with Rust 1.100 at optimization level 3, with thin
link-time optimization, a single code-generation unit, and code generation
for the native processor (\texttt{-C target-cpu=native}); the code of the
experiments is available at \url{https://github.com/vigna/phastr-bench}.

Keys are distinct
64-bit integers, obtained by multiplying their index by an odd constant and
computing the exclusive or of the result with a fixed constant; to
make query times comparable, all structures use the same hash function,
GxHash\footnote{\url{https://crates.io/crates/gxhash}} with a hash seed. \PHastR uses
$R=4$ layouts, with byte seeds for $S=8$ and bit-packed seeds read without
alignment for $S=10$. Construction times include hashing; they are
single-threaded unless otherwise noted, and they are medians of three
constructions up to $10^8$ keys and of a single one on $10^9$ keys. Query
times are medians over nine rounds of $2\cdot10^6$ queries of consecutive
keys of the set, and the rounds of the different structures are interleaved.
Since all structures hash the key first, their memory accesses are random
anyway: reading the keys in order avoids adding to every query the same cache
miss on the array of keys.

We compare \PHastR only with \PHast, \PHastP, and PtrHash~\cite{RagPTR}, the
only structures in this region of the design space. In the extensive
comparison of Beling and Sanders~\cite{BeSPHast,BeSPHastALENEX} on
$5\cdot10^8$ keys, which also includes PTHash~\cite{PTPRFMPH},
PHOBIC~\cite{HLPPPHOBSIC}, SicHash~\cite{LSWSSICTPH}, and structures based on
recursive splitting~\cite{BKLSHPCRBMPHF,LSWSTOSMPHBBF} and on
fingerprinting~\cite{BelFBMPHR}, the configurations of these structures use
between $1.8$ and $3$ bits per key and answer a query in at most about
$50$\,ns, whereas all the other structures need at least $67$\,ns, and those
using less than two bits per key at least $124$\,ns. Moreover, none of these
structures obtains, as \PHastR does, several independent placements of the
keys of a bucket from the bits left over by the computation of the bucket:
they compute a new placement for each seed, as PTHash, PHOBIC, PtrHash and
\PHast do, or they move a single placement, as hash-and-displace and \PHastP
do.

The reference implementation of \PHast and \PHastP is the \texttt{ph} crate
by Beling~\cite{BelBRLPFSDS}: we use
\texttt{Function2} with the seed choosers \texttt{ShiftOnly} for \PHastP,
\texttt{ShiftOnlyProdWrapped<$\delta$>} for \PHastP with wrapping, and
\texttt{SeedOnly(ProdOfValues)} for \PHast. The last two choose among the
feasible seeds by the product of the distances of the slots, as \PHastR does,
whereas \texttt{ShiftOnly} takes the first feasible seed. We applied a fix
to the parallel construction of \texttt{ph}: without it, all threads but the first choose
seeds using wrong bucket indices, and space grows with the number of
threads; single-threaded constructions are not affected.

We also compare with PtrHash~\cite{RagPTR} (version 2.1.2), in its two
minimal configurations: the \emph{default} one, with expected bucket size
$\lambda=3$, a linear assignment of keys to buckets, and a plain array for
the remapping, and the \emph{compact} one, with $\lambda=3.9$, a cubic
assignment of keys to buckets approximating the optimal one of
PHOBIC~\cite{HLPPPHOBSIC}, and an Elias--Fano remapping. Both store
pilots in bytes and use about $1.01n$ slots, so a small fraction of the keys
must be remapped; the default configuration is built as a single part, and
its construction is thus almost sequential, whereas the compact one is built
in parallel parts. We measure the space of PtrHash as the space of its
pilots and of its remapping.

\begin{table}[t]
\centering
\caption{Comparison on $10^7$ keys (single-threaded construction).}
\label{tab:1e7}
\small
\begin{tabular}{llrrr}
\toprule
Method & Configuration & bits/key & build [ns/key] & query [ns] \\
\midrule
\PHastP & $S=8$, $\lambda=5.25$ & 2.117 & 34 & 6.1 \\
\PHastP wrap $\delta=1$ & $S=8$, $\lambda=5.25$ & 2.043 & 52 & 5.5 \\
\PHastP wrap $\delta=2$ & $S=8$, $\lambda=5$ & 2.001 & 56 & 4.6 \\
\PHastP wrap $\delta=3$ & $S=8$, $\lambda=5$ & 1.956 & 66 & 4.5 \\
\PHastP wrap $\delta=5$ & $S=8$, $\lambda=5$ & 1.943 & 86 & 4.4 \\
\PHastP wrap $\delta=7$ & $S=8$, $\lambda=5$ & 1.940 & 106 & 4.4 \\
\PHast & $S=8$, $\lambda=4.5$ & 1.887 & 524 & 3.1 \\
PtrHash & default, $\lambda=3$ & 2.990 & 107 & 3.3 \\
PtrHash & compact, $\lambda=3.9$ & 2.143 & 122 & 5.7 \\
\addlinespace
\PHastR & $S=8$, $L=1024$, $\lambda=4.25$ & 1.943 & 38 & 2.8 \\
\PHastR & $S=8$, $L=1024$, $\lambda=4.5$ & 1.903 & 35 & 3.1 \\
\PHastR & $S=8$, $L=1024$, $\lambda=4.75$ & 1.900 & 35 & 3.6 \\
\midrule
\PHastP & $S=10$, $\lambda=6.1$ & 1.985 & 38 & 5.7 \\
\PHastP wrap $\delta=1$ & $S=10$, $\lambda=6.2$ & 1.891 & 60 & 5.0 \\
\PHastP wrap $\delta=2$ & $S=10$, $\lambda=6.25$ & 1.851 & 86 & 4.8 \\
\PHastP wrap $\delta=3$ & $S=10$, $\lambda=6$ & 1.852 & 91 & 4.6 \\
\PHastP wrap $\delta=5$ & $S=10$, $\lambda=6$ & 1.842 & 134 & 4.5 \\
\PHastP wrap $\delta=7$ & $S=10$, $\lambda=6$ & 1.841 & 172 & 4.5 \\
\PHast & $S=10$, $\lambda=5.8$ & 1.814 & 1606 & 4.5 \\
\addlinespace
\PHastR & $S=10$, $L=2048$, $\lambda=5.75$ & 1.824 & 73 & 3.4 \\
\PHastR & $S=10$, $L=2048$, $\lambda=6$ & 1.823 & 69 & 3.9 \\
\PHastR & $S=10$, $L=2048$, $\lambda=6.25$ & 1.833 & 68 & 4.3 \\
\bottomrule
\end{tabular}
\end{table}

Table~\ref{tab:1e7} shows the results on $10^7$ keys, a size at which the
structures fit in the cache. With 8-bit seeds and
the suggested expected bucket size $\lambda=4.25$ (we motivate this choice
in Section~\ref{sec:large}), \PHastR uses essentially the same space as the
best variant of \PHastP with wrapping, that with $\delta=7$, but it is built in
$36\%$ of the time, and its queries are $36\%$ faster; with $\delta=3$,
\PHastP with wrapping uses $0.013$ bits per key more, and it still takes
$1.7$ times as long to build. The queries of \PHastR are $10\%$ faster also
than those of \PHast, which takes almost fourteen times as long to build.

Queries are fast because \PHastR bumps few keys: just $0.5\%$ of the keys
from the first level, against $4.2\%$ and $3.9\%$ for \PHastP with wrapping
and $\delta=3$ and $\delta=7$, respectively, and the $1.0\%$ of \PHast (with the same expected bucket size,
however, \PHast bumps fewer keys: $0.3\%$ with $\lambda=4.25$), and queries
for bumped keys are much slower than the others. The expected bucket size thus trades space for query
time: with $\lambda=4.5$ \PHastR bumps $1.2\%$ of the keys, and it comes
within $0.02$ bits per key of the space of \PHast, with queries as fast;
with $\lambda=4.75$ it bumps $2.3\%$ of the keys, its space exceeds that of
\PHast by $0.013$ bits per key, and its queries are $16\%$ slower. \PHastP
without wrapping, instead, is slightly faster to build than \PHastR with all
three expected bucket sizes ($34$\,ns per key against $35$--$38$), but it
uses at least $0.17$ bits per key more, and its queries are much slower
($6.1$\,ns against at most $3.6$).

PtrHash, in its default configuration, has queries $18\%$ slower than those
of \PHastR with $\lambda=4.25$, it uses more than one bit per key more, and
it takes more than two and a half times as long to build; its compact
configuration uses $0.2$ bits per key more than \PHastR, and it is $3.2$
times slower to build and twice as slow to query. Note that PtrHash prioritizes
query throughput, and it supports also streams of queries, which it answers
faster than one at a time by prefetching the data needed by the following
keys: batched queries are not among the design goals of \PHastR, and we time
all structures with one query at a time.

With 10-bit seeds
and $\lambda=5.75$ \PHastR uses less space than all variants of \PHastP, and
it is built in $43\%$ of the time of \PHastP with wrapping and $\delta=7$,
and in $81\%$ of the time of that with $\delta=3$; with $\lambda=6$ it comes
within $0.01$ bits per key of the space of \PHast. In both cases its queries
are faster than those of all the other structures with 10-bit seeds.

\begin{table}[t]
\centering
\caption{Comparison on $10^8$ keys, with single-threaded and 8-threaded
construction.}
\label{tab:1e8}
\small
\begin{tabular}{llrrrr}
\toprule
& & & \multicolumn{2}{c}{build [ns/key]} & \\
\cmidrule(lr){4-5}
Method & Configuration & bits/key & 1 thread & 8 threads & query [ns] \\
\midrule
\PHastP & $S=8$, $\lambda=5.25$ & 2.116 & 36 & 7.9 & 13.6 \\
\PHastP wrap $\delta=3$ & $S=8$, $\lambda=5$ & 1.955 & 69 & 11.8 & 10.9 \\
\PHastP wrap $\delta=7$ & $S=8$, $\lambda=5$ & 1.939 & 108 & 16.7 & 10.8 \\
\PHast & $S=8$, $\lambda=4.5$ & 1.885 & 530 & 70.1 & 9.3 \\
PtrHash & default, $\lambda=3$ & 2.990 & 135 & 126.1 & 11.3 \\
PtrHash & compact, $\lambda=3.9$ & 2.143 & 125 & 22.4 & 15.1 \\
\addlinespace
\PHastR & $S=8$, $L=1024$, $\lambda=4.25$ & 1.942 & 39 & 5.8 & 8.7 \\
\PHastR & $S=8$, $L=1024$, $\lambda=4.5$ & 1.902 & 36 & 5.6 & 9.0 \\
\PHastR & $S=8$, $L=1024$, $\lambda=4.75$ & 1.898 & 35 & 5.7 & 9.6 \\
\midrule
\PHastP wrap $\delta=3$ & $S=10$, $\lambda=6$ & 1.847 & 92 & 14.9 & 11.5 \\
\PHastP wrap $\delta=7$ & $S=10$, $\lambda=6$ & 1.836 & 174 & 25.1 & 11.6 \\
\PHast & $S=10$, $\lambda=5.8$ & 1.810 & 1610 & 208.3 & 12.0 \\
\addlinespace
\PHastR & $S=10$, $L=2048$, $\lambda=5.75$ & 1.820 & 74 & 10.6 & 9.5 \\
\PHastR & $S=10$, $L=2048$, $\lambda=6$ & 1.819 & 70 & 10.2 & 9.8 \\
\bottomrule
\end{tabular}
\end{table}

Table~\ref{tab:1e8} shows that the picture does not change on $10^8$ keys:
with 8-bit seeds and
$\lambda=4.25$ \PHastR is built in $36\%$ of the time of \PHastP with
wrapping and $\delta=7$, and its queries are $19\%$ faster; they are also
$6\%$ faster than those of \PHast, and $23\%$ faster than those of PtrHash,
which is built more than three times slower. With 8 threads \PHastR is
built in $35\%$ of the time of \PHastP with wrapping and $\delta=7$ when
seeds have 8 bits, and in $42\%$ of the time when they have 10 bits (with
$\lambda=5.75$); with 8-bit seeds and 8 threads it is built faster also than
\PHastP without wrapping ($5.8$\,ns per key against $7.9$), which is instead
faster with one thread ($36$\,ns per key against $39$). The default
configuration of PtrHash, which is built as a single part, is barely faster
to build with 8 threads ($126$\,ns per key against $135$).

\subsection{Large Key Sets}
\label{sec:large}

\begin{table}[t]
\centering
\caption{Comparison on $10^9$ keys, with single-threaded and 8-threaded
construction.}
\label{tab:1e9}
\small
\begin{tabular}{llrrrr}
\toprule
& & & \multicolumn{2}{c}{build [ns/key]} & \\
\cmidrule(lr){4-5}
Method & Configuration & bits/key & 1 thread & 8 threads & query [ns] \\
\midrule
\PHastP & $S=8$, $\lambda=5.25$ & 2.115 & 40 & 8.9 & 30.4 \\
\PHastP wrap $\delta=3$ & $S=8$, $\lambda=5$ & 1.955 & 73 & 12.9 & 23.8 \\
\PHastP wrap $\delta=7$ & $S=8$, $\lambda=5$ & 1.939 & 113 & 17.9 & 23.4 \\
\PHast & $S=8$, $\lambda=4.5$ & 1.885 & 533 & 70.9 & 17.1 \\
PtrHash & default, $\lambda=3$ & 2.990 & 405 & 399.8 & 16.5 \\
PtrHash & compact, $\lambda=3.9$ & 2.143 & 125 & 23.8 & 29.3 \\
\addlinespace
\PHastR & $S=8$, $L=1024$, $\lambda=4.25$ & 1.942 & 42 & 6.2 & 15.4 \\
\PHastR & $S=8$, $L=1024$, $\lambda=4.5$ & 1.902 & 38 & 5.9 & 16.9 \\
\PHastR & $S=8$, $L=1024$, $\lambda=4.75$ & 1.897 & 37 & 5.9 & 19.3 \\
\midrule
\PHastP wrap $\delta=3$ & $S=10$, $\lambda=6$ & 1.846 & 100 & 16.3 & 25.5 \\
\PHastP wrap $\delta=7$ & $S=10$, $\lambda=6$ & 1.835 & 181 & 26.6 & 25.3 \\
\PHast & $S=10$, $\lambda=5.8$ & 1.810 & 1615 & 209.3 & 25.9 \\
\addlinespace
\PHastR & $S=10$, $L=2048$, $\lambda=5.75$ & 1.819 & 77 & 10.7 & 18.1 \\
\PHastR & $S=10$, $L=2048$, $\lambda=6$ & 1.819 & 73 & 10.3 & 20.0 \\
\bottomrule
\end{tabular}
\end{table}

Table~\ref{tab:1e9} reports the results on $10^9$ keys. Space does not depend
on the number of keys, and the ratios between the structures are similar
to those of smaller key sets: with 8-bit seeds and $\lambda=4.25$, \PHastR
is built in $37\%$ of the time of \PHastP with wrapping and $\delta=7$, and
its queries are $34\%$ faster. The exception is the default configuration of
PtrHash: its construction falls further behind that of \PHastR as the number of keys
grows (with one thread it is $2.8$, $3.4$ and $9.7$ times slower on $10^7$,
$10^8$ and $10^9$ keys, and with eight threads $16$, $22$ and $64$ times
slower), and its queries are now only $7\%$ slower than those of \PHastR, at
the price of more than one bit per key more.
With 10-bit seeds and
$\lambda=6$, \PHastR uses less than $0.01$ bits per key more than \PHast,
but it is built $22$ times faster, and its queries are $23\%$ faster. At this
scale the cost of bumped keys becomes more visible, and we used it to choose
the suggested expected bucket size.

\begin{table}[t]
\centering
\caption{Query time in nanoseconds for the keys of the first level
(``first'') and for all keys, measured in a separate run (medians over nine
rounds of $4\cdot10^6$ queries), so the times for all keys differ slightly
from those of Tables~\ref{tab:1e7} and~\ref{tab:1e9}; the second column is
the percentage of keys bumped from the first level. We omit $10^8$ keys:
at that size a structure is about as large as the L3 cache of a group of
eight cores, so query times depend on the benchmark.}
\label{tab:split}
\small
\begin{tabular}{lrrrrr}
\toprule
& & \multicolumn{2}{c}{$10^7$ keys} & \multicolumn{2}{c}{$10^9$ keys} \\
\cmidrule(lr){3-4}\cmidrule(lr){5-6}
Structure & bumped & first & all & first & all \\
\midrule
\PHastP wrap $\delta=3$, $\lambda=5$ & 4.2 & 2.3 & 4.0 & 14.4 & 23.6 \\
\PHast, $\lambda=4.5$ & 1.0 & 2.4 & 2.8 & 14.5 & 17.0 \\
\PHastR, $\lambda=4.25$ & 0.5 & 2.4 & 2.6 & 14.0 & 15.3 \\
\PHastR, $\lambda=4.5$ & 1.2 & 2.3 & 2.9 & 13.9 & 16.8 \\
\bottomrule
\end{tabular}
\end{table}

\begin{figure}[tp]
\centering
\begin{tikzpicture}
\begin{axis}[
  width=0.49\textwidth, height=0.33\textwidth, title style={font=\small},
  title={(a) queries, $10^{7}$ keys},
  xlabel={space [bits/key]}, ylabel={query time [ns]}, xmin=1.78, xmax=2.25, enlarge x limits=false, enlarge y limits=0.08, xtick={1.8,1.9,2,2.1}, extra x ticks={2.22}, extra x tick labels={$2.99$},
  grid=major, grid style={gray!20}, legend to name=leg:pareto, legend columns=3, legend cell align=left, legend style={font=\small, /tikz/every even column/.append style={column sep=0.3em}}]
\addplot[mark=*, color=black, mark size=1.8pt] coordinates { (2.123,5.7) (2.117,6.1) (2.121,6.6) };
\addplot[mark=diamond*, color=black, mark size=1.8pt] coordinates { (1.986,3.6) (1.964,4.0) (1.956,4.5) (1.966,5.0) (1.979,5.8) };
\addplot[mark=star, color=black, mark size=1.8pt] coordinates { (1.946,3.9) (1.940,4.4) (1.951,4.9) };
\addplot[mark=square*, color=blue!70!black, mark size=1.8pt] coordinates { (1.923,2.8) (1.887,3.1) (1.892,3.6) };
\addplot[mark=triangle*, color=red!70!black, mark size=1.8pt] coordinates { (2.035,2.7) (1.943,2.8) (1.903,3.1) (1.900,3.6) (1.909,4.2) };
\addplot[mark=o, color=black, mark size=1.8pt] coordinates { (1.992,5.3) (1.985,5.7) (1.993,6.1) };
\addplot[mark=diamond, color=black, mark size=1.8pt] coordinates { (1.894,4.0) (1.858,4.2) (1.852,4.6) (1.862,5.0) (1.873,5.5) };
\addplot[mark=x, color=black, mark size=1.8pt] coordinates { (1.845,4.1) (1.841,4.5) (1.852,5.0) };
\addplot[mark=square, color=blue!70!black, mark size=1.8pt] coordinates { (1.836,4.2) (1.814,4.5) (1.821,4.9) };
\addplot[mark=triangle, color=red!70!black, mark size=1.8pt] coordinates { (1.856,3.3) (1.824,3.4) (1.823,3.9) (1.833,4.3) (1.850,4.9) };
\addplot[only marks, mark=pentagon*, color=green!50!black, mark size=1.8pt] coordinates { (2.220,3.3) };
\addplot[only marks, mark=pentagon, color=green!50!black, mark size=1.8pt] coordinates { (2.143,5.7) };
\draw[densely dashed, gray] (rel axis cs:0.8723,0) -- (rel axis cs:0.8723,1);
\legend{\PHastP ($S=8$), \PHastP wrap $\delta=3$ ($S=8$), \PHastP wrap $\delta=7$ ($S=8$), \PHast ($S=8$), \PHastR ($S=8$), \PHastP ($S=10$), \PHastP wrap $\delta=3$ ($S=10$), \PHastP wrap $\delta=7$ ($S=10$), \PHast ($S=10$), \PHastR ($S=10$), PtrHash default, PtrHash compact}
\end{axis}
\end{tikzpicture}%
\begin{tikzpicture}
\begin{axis}[
  width=0.49\textwidth, height=0.33\textwidth, title style={font=\small},
  title={(b) construction, $10^{7}$ keys},
  xlabel={space [bits/key]}, ylabel={construction time [ns/key]}, ymode=log, log ticks with fixed point, ytick={20,50,100,200,500,1000,2000}, xmin=1.78, xmax=2.25, enlarge x limits=false, enlarge y limits=0.08, xtick={1.8,1.9,2,2.1}, extra x ticks={2.22}, extra x tick labels={$2.99$},
  grid=major, grid style={gray!20}]
\addplot[mark=*, color=black, mark size=1.8pt] coordinates { (2.123,33.5) (2.117,33.8) (2.121,34.2) };
\addplot[mark=diamond*, color=black, mark size=1.8pt] coordinates { (1.986,64.5) (1.964,65.3) (1.956,66.3) (1.966,67.6) (1.979,69.2) };
\addplot[mark=star, color=black, mark size=1.8pt] coordinates { (1.946,104.1) (1.940,106.2) (1.951,108.3) };
\addplot[mark=square*, color=blue!70!black, mark size=1.8pt] coordinates { (1.923,569.9) (1.887,523.7) (1.892,497.5) };
\addplot[mark=triangle*, color=red!70!black, mark size=1.8pt] coordinates { (2.035,43.9) (1.943,38.4) (1.903,35.2) (1.900,34.6) (1.909,34.1) };
\addplot[mark=o, color=black, mark size=1.8pt] coordinates { (1.992,38.0) (1.985,38.2) (1.993,38.9) };
\addplot[mark=diamond, color=black, mark size=1.8pt] coordinates { (1.894,85.5) (1.858,88.4) (1.852,90.6) (1.862,92.6) (1.873,94.5) };
\addplot[mark=x, color=black, mark size=1.8pt] coordinates { (1.845,167.4) (1.841,172.3) (1.852,176.0) };
\addplot[mark=square, color=blue!70!black, mark size=1.8pt] coordinates { (1.836,1730.8) (1.814,1605.6) (1.821,1523.9) };
\addplot[mark=triangle, color=red!70!black, mark size=1.8pt] coordinates { (1.856,82.0) (1.824,73.3) (1.823,69.0) (1.833,67.9) (1.850,67.6) };
\addplot[only marks, mark=pentagon*, color=green!50!black, mark size=1.8pt] coordinates { (2.220,106.7) };
\addplot[only marks, mark=pentagon, color=green!50!black, mark size=1.8pt] coordinates { (2.143,122.2) };
\draw[densely dashed, gray] (rel axis cs:0.8723,0) -- (rel axis cs:0.8723,1);
\end{axis}
\end{tikzpicture}\\[1ex]
\begin{tikzpicture}
\begin{axis}[
  width=0.49\textwidth, height=0.33\textwidth, title style={font=\small},
  title={(c) queries, $10^{8}$ keys},
  xlabel={space [bits/key]}, ylabel={query time [ns]}, xmin=1.78, xmax=2.25, enlarge x limits=false, enlarge y limits=0.08, xtick={1.8,1.9,2,2.1}, extra x ticks={2.22}, extra x tick labels={$2.99$},
  grid=major, grid style={gray!20}]
\addplot[mark=*, color=black, mark size=1.8pt] coordinates { (2.123,12.9) (2.116,13.6) (2.120,14.3) };
\addplot[mark=diamond*, color=black, mark size=1.8pt] coordinates { (1.985,9.9) (1.962,10.6) (1.955,10.9) (1.964,11.5) (1.978,12.6) };
\addplot[mark=star, color=black, mark size=1.8pt] coordinates { (1.944,10.3) (1.939,10.8) (1.950,11.5) };
\addplot[mark=square*, color=blue!70!black, mark size=1.8pt] coordinates { (1.922,8.9) (1.885,9.3) (1.889,9.9) };
\addplot[mark=triangle*, color=red!70!black, mark size=1.8pt] coordinates { (2.035,9.0) (1.942,8.7) (1.902,9.0) (1.898,9.6) (1.908,10.1) };
\addplot[mark=o, color=black, mark size=1.8pt] coordinates { (1.988,13.6) (1.981,14.2) (1.988,14.5) };
\addplot[mark=diamond, color=black, mark size=1.8pt] coordinates { (1.890,11.5) (1.853,11.3) (1.847,11.5) (1.856,12.4) (1.867,13.2) };
\addplot[mark=x, color=black, mark size=1.8pt] coordinates { (1.840,11.3) (1.836,11.6) (1.846,12.3) };
\addplot[mark=square, color=blue!70!black, mark size=1.8pt] coordinates { (1.834,11.5) (1.810,12.0) (1.816,12.5) };
\addplot[mark=triangle, color=red!70!black, mark size=1.8pt] coordinates { (1.854,9.4) (1.820,9.5) (1.819,9.8) (1.829,10.3) (1.846,11.2) };
\addplot[only marks, mark=pentagon*, color=green!50!black, mark size=1.8pt] coordinates { (2.220,11.3) };
\addplot[only marks, mark=pentagon, color=green!50!black, mark size=1.8pt] coordinates { (2.143,15.1) };
\draw[densely dashed, gray] (rel axis cs:0.8723,0) -- (rel axis cs:0.8723,1);
\end{axis}
\end{tikzpicture}%
\begin{tikzpicture}
\begin{axis}[
  width=0.49\textwidth, height=0.33\textwidth, title style={font=\small},
  title={(d) construction, $10^{8}$ keys},
  xlabel={space [bits/key]}, ylabel={construction time [ns/key]}, ymode=log, log ticks with fixed point, ytick={20,50,100,200,500,1000,2000}, xmin=1.78, xmax=2.25, enlarge x limits=false, enlarge y limits=0.08, xtick={1.8,1.9,2,2.1}, extra x ticks={2.22}, extra x tick labels={$2.99$},
  grid=major, grid style={gray!20}]
\addplot[mark=*, color=black, mark size=1.8pt] coordinates { (2.123,35.7) (2.116,35.8) (2.120,36.2) };
\addplot[mark=diamond*, color=black, mark size=1.8pt] coordinates { (1.985,66.8) (1.962,68.0) (1.955,68.9) (1.964,69.4) (1.978,70.6) };
\addplot[mark=star, color=black, mark size=1.8pt] coordinates { (1.944,107.0) (1.939,108.4) (1.950,110.6) };
\addplot[mark=square*, color=blue!70!black, mark size=1.8pt] coordinates { (1.922,570.8) (1.885,530.3) (1.889,501.0) };
\addplot[mark=triangle*, color=red!70!black, mark size=1.8pt] coordinates { (2.035,45.1) (1.942,39.3) (1.902,36.1) (1.898,35.0) (1.908,34.9) };
\addplot[mark=o, color=black, mark size=1.8pt] coordinates { (1.988,40.3) (1.981,40.7) (1.988,41.0) };
\addplot[mark=diamond, color=black, mark size=1.8pt] coordinates { (1.890,87.8) (1.853,90.1) (1.847,92.2) (1.856,94.4) (1.867,96.3) };
\addplot[mark=x, color=black, mark size=1.8pt] coordinates { (1.840,169.5) (1.836,174.0) (1.846,178.0) };
\addplot[mark=square, color=blue!70!black, mark size=1.8pt] coordinates { (1.834,1739.3) (1.810,1610.1) (1.816,1528.0) };
\addplot[mark=triangle, color=red!70!black, mark size=1.8pt] coordinates { (1.854,83.6) (1.820,74.1) (1.819,69.9) (1.829,68.6) (1.846,68.5) };
\addplot[only marks, mark=pentagon*, color=green!50!black, mark size=1.8pt] coordinates { (2.220,135.0) };
\addplot[only marks, mark=pentagon, color=green!50!black, mark size=1.8pt] coordinates { (2.143,124.9) };
\draw[densely dashed, gray] (rel axis cs:0.8723,0) -- (rel axis cs:0.8723,1);
\end{axis}
\end{tikzpicture}\\[1ex]
\begin{tikzpicture}
\begin{axis}[
  width=0.49\textwidth, height=0.33\textwidth, title style={font=\small},
  title={(e) queries, $10^{9}$ keys},
  xlabel={space [bits/key]}, ylabel={query time [ns]}, xmin=1.78, xmax=2.25, enlarge x limits=false, enlarge y limits=0.08, xtick={1.8,1.9,2,2.1}, extra x ticks={2.22}, extra x tick labels={$2.99$},
  grid=major, grid style={gray!20}]
\addplot[mark=*, color=black, mark size=1.8pt] coordinates { (2.122,28.6) (2.115,30.4) (2.120,32.4) };
\addplot[mark=diamond*, color=black, mark size=1.8pt] coordinates { (1.984,19.7) (1.962,21.7) (1.955,23.8) (1.964,26.1) (1.978,29.0) };
\addplot[mark=star, color=black, mark size=1.8pt] coordinates { (1.944,21.0) (1.939,23.4) (1.950,25.6) };
\addplot[mark=square*, color=blue!70!black, mark size=1.8pt] coordinates { (1.922,15.4) (1.885,17.1) (1.889,19.6) };
\addplot[mark=triangle*, color=red!70!black, mark size=1.8pt] coordinates { (2.034,14.9) (1.942,15.4) (1.902,16.9) (1.897,19.3) (1.908,21.9) };
\addplot[mark=o, color=black, mark size=1.8pt] coordinates { (1.987,28.9) (1.981,30.9) (1.987,32.8) };
\addplot[mark=diamond, color=black, mark size=1.8pt] coordinates { (1.889,22.6) (1.852,23.8) (1.846,25.5) (1.855,27.9) (1.867,30.2) };
\addplot[mark=x, color=black, mark size=1.8pt] coordinates { (1.840,23.3) (1.835,25.3) (1.845,27.5) };
\addplot[mark=square, color=blue!70!black, mark size=1.8pt] coordinates { (1.833,24.5) (1.810,25.9) (1.816,28.0) };
\addplot[mark=triangle, color=red!70!black, mark size=1.8pt] coordinates { (1.853,17.0) (1.819,18.1) (1.819,20.0) (1.829,22.2) (1.846,24.9) };
\addplot[only marks, mark=pentagon*, color=green!50!black, mark size=1.8pt] coordinates { (2.220,16.5) };
\addplot[only marks, mark=pentagon, color=green!50!black, mark size=1.8pt] coordinates { (2.143,29.3) };
\draw[densely dashed, gray] (rel axis cs:0.8723,0) -- (rel axis cs:0.8723,1);
\end{axis}
\end{tikzpicture}%
\begin{tikzpicture}
\begin{axis}[
  width=0.49\textwidth, height=0.33\textwidth, title style={font=\small},
  title={(f) construction, $10^{9}$ keys},
  xlabel={space [bits/key]}, ylabel={construction time [ns/key]}, ymode=log, log ticks with fixed point, ytick={20,50,100,200,500,1000,2000}, xmin=1.78, xmax=2.25, enlarge x limits=false, enlarge y limits=0.08, xtick={1.8,1.9,2,2.1}, extra x ticks={2.22}, extra x tick labels={$2.99$},
  grid=major, grid style={gray!20}]
\addplot[mark=*, color=black, mark size=1.8pt] coordinates { (2.122,40.2) (2.115,40.4) (2.120,40.7) };
\addplot[mark=diamond*, color=black, mark size=1.8pt] coordinates { (1.984,70.8) (1.962,72.2) (1.955,73.1) (1.964,74.2) (1.978,75.7) };
\addplot[mark=star, color=black, mark size=1.8pt] coordinates { (1.944,111.0) (1.939,112.8) (1.950,114.9) };
\addplot[mark=square*, color=blue!70!black, mark size=1.8pt] coordinates { (1.922,574.7) (1.885,532.9) (1.889,506.0) };
\addplot[mark=triangle*, color=red!70!black, mark size=1.8pt] coordinates { (2.034,47.8) (1.942,41.6) (1.902,37.8) (1.897,37.0) (1.908,37.2) };
\addplot[mark=o, color=black, mark size=1.8pt] coordinates { (1.987,47.0) (1.981,47.3) (1.987,47.9) };
\addplot[mark=diamond, color=black, mark size=1.8pt] coordinates { (1.889,94.1) (1.852,97.3) (1.846,99.5) (1.855,101.7) (1.867,103.2) };
\addplot[mark=x, color=black, mark size=1.8pt] coordinates { (1.840,176.0) (1.835,181.2) (1.845,184.7) };
\addplot[mark=square, color=blue!70!black, mark size=1.8pt] coordinates { (1.833,1748.2) (1.810,1615.2) (1.816,1533.8) };
\addplot[mark=triangle, color=red!70!black, mark size=1.8pt] coordinates { (1.853,85.8) (1.819,77.1) (1.819,73.1) (1.829,71.7) (1.846,71.6) };
\addplot[only marks, mark=pentagon*, color=green!50!black, mark size=1.8pt] coordinates { (2.220,405.1) };
\addplot[only marks, mark=pentagon, color=green!50!black, mark size=1.8pt] coordinates { (2.143,125.0) };
\draw[densely dashed, gray] (rel axis cs:0.8723,0) -- (rel axis cs:0.8723,1);
\end{axis}
\end{tikzpicture}\\[1ex]
\ref{leg:pareto}
\caption{Space versus query time (left) and single-threaded construction time
(right) on $10^7$, $10^8$ and $10^9$ keys (top to bottom) as the expected
bucket size varies in steps of $0.25$ around the configurations of
Table~\ref{tab:1e7}: from $4$ to $5$ (\PHastR), from $4.5$ to $5.5$
(\PHastP wrap $\delta=3$) and by $0.25$ on each side (the others) with 8-bit
seeds, and from $5.5$ to $6.5$ (\PHastR and \PHastP wrap $\delta=3$) and by
$0.25$ on each side (the others) with 10-bit seeds; the configurations of
each structure are joined in order of expected bucket size. \PHastP with
wrapping is shown only with $\delta=3$ and $\delta=7$, and the two
configurations of PtrHash are single points; the space axis is broken at the
dashed line, beyond which the default configuration of PtrHash, which uses
$2.99$ bits per key, is drawn. Each point is a configuration: lower and
further to the left is better.}
\label{fig:pareto}
\end{figure}
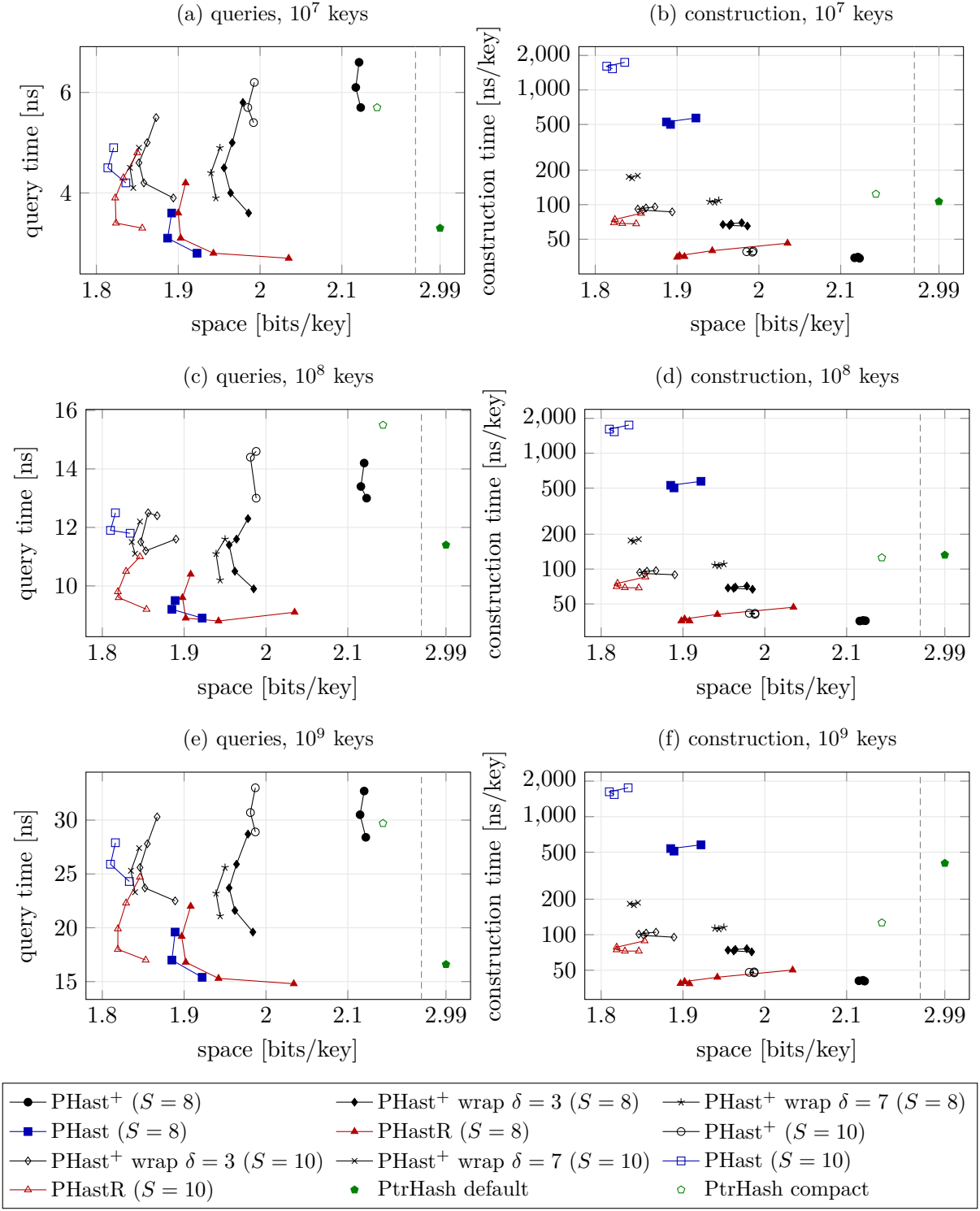

\paragraph{Trade-offs.}
Figure~\ref{fig:pareto} shows what happens when the expected bucket size
varies around the configurations of the tables. Space has a minimum, query
time grows with $\lambda$, as more keys are bumped, and construction time
varies by at most $15\%$, except for \PHastR, whose construction becomes
slower as $\lambda$ decreases (by up to $29\%$). The curves of \PHastR are
slightly to the right of those of \PHast (their minima differ by about
$0.01$ bits per key), but their construction times differ by more than a
factor of ten; at about the same space, with 8-bit seeds the queries of the
two structures are about as fast, whereas with 10-bit seeds the curves of
\PHastR are below those of \PHast, as its queries are faster.

The curves of \PHastP with
wrapping are to the right of both, even with $\delta=7$. Considering the
three dimensions at once, almost every configuration of \PHastP, with or
without wrapping (also with $\delta=1$, $2$ and $5$, which are not plotted),
is dominated by a configuration of \PHastR with less space, faster queries
and a faster construction, and every configuration of \PHast or \PHastP
that is not dominated wins in one dimension by a small margin, but loses in
another by a large one: some configurations of \PHast use $0.01$--$0.02$
bits per key less, or have queries less than $1$\,ns faster, but they take
more than ten times as long to build; \PHastP with wrapping, $\delta=1$ and
10-bit seeds uses less than $0.01$ bits per key less than \PHastR with 8-bit
seeds and $\lambda=4.75$, but it is built at least $70\%$ slower, and its
queries are more than $20\%$ slower; and \PHastP without wrapping, with
8-bit seeds, is built less than $1$\,ns per key faster on $10^7$ keys, but
it uses more than $0.2$ bits per key more, and its queries are at least
$36\%$ slower.

\paragraph{Bumped keys.}
Table~\ref{tab:split} shows the time of queries for the keys of the first
level only, for four of the configurations of the previous tables. It is
essentially the same for \PHastP with wrapping and for \PHastR, and it grows
with the number of keys, as the array of seeds becomes larger; it is slightly larger for the reference implementation of \PHast,
which explains part of its gap. Besides that, the differences between the
structures are due to bumped keys. In a workload in which all keys are queried, a query
for a bumped key takes about $45$\,ns more than a query for a key of the
first level on $10^7$ keys, and about $220$--$260$\,ns more on $10^9$ keys: a
bumped key reads the seed of another level, and then the remapping, whose
accesses (the inventory of the selection structure, and then the upper bits of the Elias--Fano
sequence) are dependent cache misses. Thus, each percent of bumped keys
costs on average up to half a nanosecond per query on $10^7$ keys, and
more than $2$\,ns on $10^9$ keys: the larger the key set, the more a small
expected bucket size is convenient. This is the reason why we suggest $\lambda=4.25$ (and, for the same reason, $\lambda=5.75$ with 10-bit seeds): with
$\lambda=4.5$ \PHastR uses $0.04$ bits per key less, and it is built faster,
but on $10^9$ keys its queries are $10\%$ slower.

\begin{table}[t]
\centering
\caption{\PHastR with $1$, $2$ and $4$ layouts on $10^9$ keys
(single-threaded construction), with the suggested expected bucket size and
with the one that minimizes space for each number of layouts. Queries are
timed with the code used for all numbers of layouts (generic) and, for four
layouts, also with the code that \PHastR uses by default (specialized).}
\label{tab:layouts}
\small
\begin{tabular}{lrrrrr}
\toprule
& & & & \multicolumn{2}{c}{query [ns]} \\
\cmidrule(lr){5-6}
Configuration & bits/key & bumped [\%] & build [ns/key] & generic & specialized \\
\midrule
$S=8$, $L=1024$, $\lambda=4.25$, $R=1$ & 2.060 & 1.79 & 55 & 19.6 & -- \\
$S=8$, $L=1024$, $\lambda=4.25$, $R=2$ & 1.952 & 0.62 & 54 & 16.7 & -- \\
$S=8$, $L=1024$, $\lambda=4.25$, $R=4$ & 1.942 & 0.51 & 42 & 16.4 & 15.4 \\
\addlinespace
$S=8$, $L=1024$, $\lambda=5$, $R=1$ & 1.952 & 4.12 & 47 & 24.8 & -- \\
$S=8$, $L=1024$, $\lambda=4.75$, $R=2$ & 1.906 & 2.38 & 48 & 20.7 & -- \\
$S=8$, $L=1024$, $\lambda=4.75$, $R=4$ & 1.897 & 2.27 & 38 & 20.5 & 19.3 \\
\midrule
$S=10$, $L=2048$, $\lambda=5.75$, $R=1$ & 1.876 & 1.35 & 73 & 22.3 & -- \\
$S=10$, $L=2048$, $\lambda=5.75$, $R=2$ & 1.824 & 0.78 & 72 & 20.9 & -- \\
$S=10$, $L=2048$, $\lambda=5.75$, $R=4$ & 1.819 & 0.73 & 77 & 20.7 & 18.0 \\
\addlinespace
$S=10$, $L=2048$, $\lambda=6.25$, $R=1$ & 1.849 & 2.73 & 66 & 25.7 & -- \\
$S=10$, $L=2048$, $\lambda=6$, $R=2$ & 1.821 & 1.56 & 69 & 22.9 & -- \\
$S=10$, $L=2048$, $\lambda=6$, $R=4$ & 1.819 & 1.52 & 73 & 22.6 & 20.0 \\
\bottomrule
\end{tabular}
\end{table}

\paragraph{Number of layouts.}
Table~\ref{tab:layouts} isolates the effect of layouts, comparing
structures that differ only in their number. With 8-bit seeds and the
suggested expected bucket size, going from one to four layouts reduces the
bumped keys from $1.8\%$ to $0.5\%$ and saves $0.12$ bits per key; queries
become $16\%$ faster with the same code, and another $6\%$ faster with the
code that \PHastR uses by default, which is specialized for four layouts
(with 8-bit seeds, it reads the field of the layout with a 16-bit load; see
Section~\ref{sec:phastr}).
The advantage persists when each number of layouts uses the expected bucket
size that minimizes its space: four layouts still save $0.055$ bits per key
with respect to one, and they are faster both to build and to query. Four
layouts are also the fastest to build: only with four layouts do the bits of the
ring of a key fill a single word (Section~\ref{sec:phastr}), and, as for
queries, the construction uses code specialized for this configuration. With 10-bit seeds the
reduction is smaller (from $1.4\%$ to $0.7\%$ of bumped keys, saving $0.06$
bits per key, and $0.03$ when each number of layouts uses its best expected
bucket size), and the construction with four layouts is $6$--$7\%$ slower
than with two.

Self-collisions explain most of the reduction: on $10^7$ keys, with $S=8$ and a single
layout, the buckets containing a self-collision hold $1.28\%$ of the keys, in
agreement with the first-order estimate $\lambda(\lambda+2)/2L\approx
1.30\%$ for Poisson bucket sizes, and they contain $71\%$ of the bumped keys; with two layouts, the
buckets with a self-collision in both layouts hold just $0.03\%$ of the keys,
and with four layouts there are none (the estimate of
Section~\ref{sec:phastr} predicts fewer than one), so the further reduction from two to
four layouts is due to the larger number of independent placements.

\begin{table}[t]
\centering
\caption{Construction on $10^9$ keys ($S=8$): time in nanoseconds per key by
number of threads (each thread runs on a distinct core), speedup with
$64$ threads, and peak memory usage (resident set size) with eight threads in
bytes per key, excluding the keys (not measured for PtrHash). The times were
measured in a separate run, so those with one thread differ slightly from
those of Table~\ref{tab:1e9}.}
\label{tab:threads}
\small
\setlength{\tabcolsep}{3.5pt}
\begin{tabular}{lrrrrrrrrrr}
\toprule
& \multicolumn{8}{c}{threads} & & \\
\cmidrule(lr){2-9}
Structure & 1 & 2 & 4 & 8 & 16 & 32 & 64 & 96 & speedup & bytes/key \\
\midrule
\PHastP, $\lambda=5.25$ & 40.2 & 23.0 & 13.3 & 8.9 & 6.3 & 5.0 & 4.7 & 4.6 & 8.6 & 9.9 \\
\PHastP wrap $\delta=3$, $\lambda=5$ & 73.0 & 39.5 & 21.4 & 12.9 & 8.4 & 6.1 & 5.0 & 4.9 & 14.6 & 9.9 \\
\PHast, $\lambda=4.5$ & 533.7 & 269.5 & 136.5 & 70.9 & 37.2 & 20.4 & 12.9 & 10.8 & 41.4 & 10.1 \\
PtrHash, default & 404.3 & 398.0 & 397.6 & 397.6 & 400.7 & 396.7 & 399.0 & 398.1 & 1.0 & -- \\
PtrHash, compact & 125.2 & 65.7 & 36.6 & 23.6 & 12.3 & 6.5 & 4.0 & 3.2 & 31.3 & -- \\
\addlinespace
\PHastR, $\lambda=4.25$ & 42.1 & 21.5 & 11.3 & 6.3 & 3.6 & 2.3 & 1.9 & 1.7 & 22.2 & 10.6 \\
\bottomrule
\end{tabular}
\end{table}

\paragraph{Construction.}
The construction time per key of \PHastR grows slowly with the number of
keys ($38$, $39$ and $42$\,ns with one thread on $10^7$, $10^8$ and $10^9$
keys, respectively), as parts are distributed into smaller parts that fit
the cache whatever the number of keys (Section~\ref{sec:phastr}); the growth
is partly due to the walk finding the keys of bumped buckets, which randomly
accesses the bit vector of bumped buckets.

Table~\ref{tab:threads} shows how
construction scales with the number of threads: with $64$ threads \PHastR is
built $22.2$ times faster than with one, whereas \PHastP is built $8.6$ times
faster, and \PHastP with wrapping $14.6$ times faster; the construction of
\PHastP stops getting faster at about $4.6$\,ns per key. With all the
cores of the machine, the structure on $10^9$ keys is built in $1.7$\,ns
per key, $2.7$ and $2.9$ times faster than \PHastP without and with
wrapping, respectively; \PHast, whose construction time is dominated by the
search of seeds, has the largest speedup, but it is still more than six
times slower.

The compact configuration of PtrHash scales well, too ($31.3$ times faster
with $64$ threads), but with all threads it is still $1.9$ times slower than
\PHastR; the default one, which is built as a single part, does not scale.
The loss with respect to a linear speedup comes from the lower clock
frequency when many cores are active, from memory contention, and from the time spent by the
kernel serving page faults. The peak memory usage is close to that of
the reference implementation of \PHastP, which sorts hashes in place and
hashes all keys again to find those of bumped buckets.

\section{Offline Construction}
\label{sec:offline}

The construction of Section~\ref{sec:phastr} keeps in memory the hashes of all
keys and, as bumped keys are hashed again, the keys themselves: about $11$
bytes per key besides the keys. To build structures on hundreds of billions
of keys, \PHastR has also an offline construction, which reads the keys just
once, keeps their hashes on disk, needs in memory $256$\,MiB of buffers
while it reads the keys and then less than half a byte per key, and builds
exactly the same structure. 

The external-memory construction of
PTHash~\cite{PTPEMCMPHFP} also builds the same structure as its construction
in memory, but since PTHash processes buckets in order of decreasing size, it
must sort the hashes by bucket on disk and group the buckets by size. Beling
and Sanders~\cite{BeSPHast,BeSPHastALENEX} outline, but do not implement, an external-memory
construction of \PHast, which also sorts the hashes by bucket on disk, and
reads the keys twice, the second time to find the keys of bumped buckets.

In
\PHastR, as in the window of \PHast, buckets enter the sweep in order of
index, so sorting the hashes by bucket requires just distributing them on
disk into partitions, that is, the first pass of a bucket sort; each
partition is then sorted in memory when the sweep reaches it. Distributing the hashes
on disk by their highest bits has been a standard technique of scalable
constructions of static functions and MPHFs since the work of Botelho et al.~\cite{BPZPPHNOS,BZEPHVLKS}, but
these constructions build a separate structure for each partition, whereas
in \PHastR the linear map from hashes to buckets and the sweep in order of
index make it possible to build in this way a single structure, identical to
the one built in memory.

\paragraph{Records.}
The hash $h_i$ of a key at every level depends only on its hash~$h$ and on
its second hash~$h'$ (Section~\ref{sec:phastr}).
Thus, while reading the keys we compute both hashes of every key, and write a
$16$-byte \emph{record} $(h,h')$ to temporary files: the keys are no longer
necessary, and they can be streamed, for example, from a file. The records of
a level are distributed into a power-of-two number of \emph{partitions} by
the upper bits of the hash of the key at the level. Each partition is a sequence of
blocks, which are appended as their buffers fill to one of at most $16$ files,
each containing the blocks of a range of consecutive partitions, and a table
locates the blocks of each partition.

\paragraph{Parts.}
Since $\hi(hB)$ is a nondecreasing function of $h$, the keys of a partition
belong to consecutive buckets, and only the first and the last of these
buckets can contain also keys of other partitions. Offline, the parts
(Section~\ref{sec:phastr}), which sweeps load in increasing order, are made
of consecutive partitions containing together at least $512$ buckets, and a
bucket whose keys lie in two parts belongs to the second one: when a sweep
loads a part, it sets aside the keys found in the part that belong to the
first bucket of the following part, and adds them to that part when it loads
it. Thus, each partition is read once by each sweep, except at the beginning
of a sweep, where the previous part might need to be read too.

\paragraph{Same structure.}
Given the slots already used, the seed chosen for a bucket depends only on
the set of the hashes of its keys, and not on their order: free rotations are
computed by a bitwise \textsc{or}, the cost of a rotation is a sum of
integers read from a table, and ties are broken by a fixed rule
(Section~\ref{sec:phastr}). Since the sweeps see the same
buckets in the same order, the structure built offline is identical to that
built in memory, provided that the chunks of parallel sweeps are the same
(Section~\ref{sec:phastr}). The only difference is in
the detection of duplicate keys (Section~\ref{sec:phastr}): as the keys are no
longer available for a third hash, offline we consider duplicates two keys
with the same record. The two criteria differ only in the presence of two
distinct keys with the same $128$ bits of hashes, in which case both
constructions fail.

\paragraph{Following levels and memory.}
After the sweep of a level, its files are read again, and the records of the
keys of bumped buckets are written to the files of the following
level, or kept in memory if they are at most $4096$: the files of the first
level are written once and read twice, and the following ones are much
smaller. In memory, besides the buffers of the partitions it reads, the sweep of the first level needs its seeds, as bytes when they
have at most eight bits, and its set of used slots, which the sweeps of a
parallel construction update directly, as each sets only the bits of its
slots: with $S=8$ and $\lambda=4.25$, about $1/4.25+1/8\approx0.36$ bytes
per key. Then, the seeds of the first level become those of the structure,
its free slots are kept in an Elias--Fano sequence, and the remapping is
built at the end from the sets of used slots of the following levels, which
are small: from this point on, memory is essentially that of the structure
being built.

\paragraph{Experiments.}
On $10^9$ keys the offline construction allocates at most $0.38$ bytes per
key with one thread and $0.41$ with eight, against $9.8$ and $11.4$ bytes
per key, excluding the keys, of the construction in memory (as measured by a
counting allocator, which counts also memory that is allocated but never
used: the resident set size of the construction in memory with eight
threads is $10.6$ bytes per key, Table~\ref{tab:threads}). It takes $46.7$\,ns per key with one thread and $15.8$
with eight, against $42.4$ and $6.2$ of the construction in memory in the
same run: with one thread the difference is small, but with eight the time
is dominated by the first phase, in which a single thread reads the keys
and writes the records. On
$10^{11}$ keys, the offline construction would need about $40$\,GB of memory
and $1.6$\,TB of disk,
whereas the construction in memory would need about $1$\,TB of memory
besides the keys.

\paragraph{A trillion keys.} We built offline a structure on the $10^{12}$
$64$-bit integers $[0\..10^{12})$, with $S=8$ and $\lambda=4.25$, using the
command-line tool of \texttt{sux} with $48$ threads. We ran this experiment on
an Amazon EC2 \texttt{r8a.metal-24xl} instance, which has the same processor as
the instance of Section~\ref{sec:experiments} and $768$\,GiB of RAM, keeping the
temporary files ($16$\,TB) on a single EBS \texttt{gp3} volume of $20$\,TB,
provisioned with a throughput of $2000$\,MiB/s and $16\,000$ IOPS. Construction
took $9.3$ hours ($33.4$\,ns per key), $3.7$ of which in the first phase, with a
peak resident set size of $463$\,GB ($0.46$ bytes per key). About $2.7\cdot10^4$
pairs of keys have the same $64$-bit hash, which is negligible
(Section~\ref{sec:phastr}): the first level bumped $0.512\%$ of the keys,
against $0.513\%$ on $10^9$ keys in the same setting, and the structure takes
$1.942$ bits per key. Queries on the structure loaded in memory take $48.5$\,ns,
against $15.4$ on $10^9$ keys (Table~\ref{tab:1e9}).

\section{The Limits of the Approach}
\label{sec:limits}

The first level of \PHastR still bumps between $0.5\%$ and $1.2\%$ of the
keys with 8-bit seeds and $\lambda=4.25$ or $\lambda=4.5$, and, once the
seed width and the expected bucket size are fixed, bumping is where the
excess space goes: it is natural to ask how much a more flexible placement
scheme, or a more powerful search, such as the backtracking and the local
search suggested by Beling and Sanders~\cite{BeSPHast,BeSPHastALENEX}, could reduce it. In
this section we argue that on large key sets bumping cannot be avoided, for a
structural reason that applies to every placement scheme of the \PHast
family, and we quantify it: we derive an asymptotic lower bound of about
$0.1\,n/W$ on the expected number of bumped keys, and we prove that, in a
Poisson model, the best placement bumps about $n/2W$ keys, that is, a
fraction $1/2W$ of the keys, where $W$ is the
number of slots in which a key can be placed. For \PHastR this is about one tenth of the keys that it bumps.

In all structures of the \PHast family a key whose slice starts at $\sigma$
can be placed only in its \emph{reach} $[\sigma\..\sigma+W)$, where $W=L$ for
\PHast, for \PHastP with wrapping, and for \PHastR, and $W=L+2^S-2$ for
\PHastP without wrapping. The seed of a bucket chooses the slots of its keys
within their reaches, but the reaches themselves depend only on the hashes of
the keys. This simple constraint is sufficient to force bumping.

\begin{proposition}
\label{prop:bump}
Consider $n$ keys and $n$ slots, and assume that, for some $W\geq 1$, each key
$x$ can be placed only in the slots $[\sigma_x\..\sigma_x+W)$. Let
\[
X(t) = t - \bigl|\{\,x \mid \sigma_x < t\,\}\bigr|\qquad (0\leq t\leq n),
\]
and let $D=\max_t X(t)-\min_t X(t)$ be the range of $X$. Then, every injective
placement of a subset of the keys bumps (i.e., leaves out) at least $D-W+1$
keys.
\end{proposition}

\begin{proof}
For $0\leq t\leq t'\leq n$ the number of keys $x$ with $t\leq\sigma_x<t'$ is
$t'-t-\bigl(X(t')-X(t)\bigr)$. Only these keys can be placed in the $t'-t-W+1$
slots in $[t+W-1\..t')$ (which, as $W\geq1$, lie in $[t\..t')$), so at least
$X(t')-X(t)-W+1$ of these slots are free.
Symmetrically, these keys can be placed only in the at most $t'-t+W-1$ slots
in $[t\..t'+W-1)$, so at least $X(t)-X(t')-W+1$ of them are bumped. Since there
are as many keys as slots, the number of bumped keys is equal to the number of
free slots, and it is thus at least $\bigl|X(t')-X(t)\bigr|-W+1$ for all
$t\leq t'$. Choosing as $t$ and $t'$ the points at which $X$ attains its
minimum and its maximum, in the appropriate order, completes the proof.
\end{proof}

The proposition is a quantitative version of the (easy) necessity of Hall's
condition~\cite{HalORS} for the bipartite graph connecting each key to the
slots of its reach. In the \PHast family slice starts are computed from a
hash by a multiplication in fixed-point arithmetic, so we can model them as
independent and uniformly distributed. If the $\sigma_x$ are independent and
uniform in $[0\..n)$, then $D=nV_n-1$, where $V_n$ is Kuiper's
statistic~\cite{KuiTCRPC} of the values $\sigma_x/n$, that is, the sum of the
maximum positive and negative deviations of their empirical distribution from
the uniform distribution (the compression of slice starts into
$[0\..n-W+1)$ changes $D$ by at most $W-1$, which does not affect the
asymptotics).

By Donsker's theorem~\cite{BilCPM} and the continuity of the
range, $D/\sqrt n$ converges in distribution to the range of a standard
Brownian bridge (which, by Vervaat's transformation~\cite{VerRBBBE}, has the
same distribution as the maximum of a standard Brownian
excursion~\cite{KenDMBE}). By symmetry, the mean of the range is twice the
mean of the maximum of the bridge, which exceeds $x\geq0$ with probability
$e^{-2x^2}$~\cite{BilCPM}, so it is $2\int_0^\infty
e^{-2x^2}\,dx=\sqrt{\pi/2}$. Since $V_n$ is at most twice the maximum absolute
deviation of the empirical distribution from the uniform distribution, the
Dvoretzky--Kiefer--Wolfowitz inequality~\cite{MasTCDKWI} shows that the tails
of $D/\sqrt n$ are sub-Gaussian, uniformly in $n$; thus, the convergence
extends to expectations, and we conclude that
\[
\mathbf E[D]=\sqrt{\frac{\pi n}2}\bigl(1+o(1)\bigr).
\]

\begin{table}[t]
\centering
\caption{Percentage of keys that every placement must bump when each key can
be placed only in the $W=1024$ slots of its slice, for the slice starts of the
first level of \PHastR with $S=8$ on our key sets: the bound $D-W+1$ of
Proposition~\ref{prop:bump}, the minimum (computed by the greedy placement of
Proposition~\ref{prop:greedy}),
and the value $n/2W$ given by the Poisson model of
Proposition~\ref{prop:poisson}, compared with the keys bumped by \PHastR with
$\lambda=4.25$.}
\label{tab:floor}
\small
\begin{tabular}{lrrrr}
\toprule
Keys & $D-W+1$ & minimum & $n/2W$ & \PHastR \\
\midrule
$10^7$ & 0.023 & 0.046 & 0.049 & 0.51 \\
$10^8$ & 0.013 & 0.053 & 0.049 & 0.52 \\
$10^9$ & 0.003 & 0.050 & 0.049 & 0.51 \\
\bottomrule
\end{tabular}
\end{table}

\begin{figure}[t]
\centering
\begin{tikzpicture}[x=3.8mm,y=4mm,font=\footnotesize,>={Latex[length=1.3mm]}]
\fill[gray!12] (7,0.6) rectangle (14,10.95);
\fill[gray!12] (22,0.6) rectangle (28,10.95);
\node[left] at (-.2,11.6) {slots};
\foreach \i in {0,...,31} \draw (\i,11.2) rectangle ++(1,.8);
\draw[decorate,decoration={brace,amplitude=4pt}] (7,12.2) -- (17,12.2) node[midway,above=4pt,font=\scriptsize] {10 slots};
\draw[decorate,decoration={brace,amplitude=4pt}] (22,12.2) -- (31,12.2) node[midway,above=4pt,font=\scriptsize] {9 slots};
\node[left] at (-.2,8.15) {reaches};
\foreach \s/\y in {0/10.4,1/9.7,14/6.2,15/10.4,19/10.4,28/7.6}
  \draw[gray!60,fill=gray!8] (\s+.08,\y) rectangle (\s+3.92,\y+.45);
\foreach \s/\y in {7/9.7,7/10.4,8/7.6,8/8.3,8/9,9/6.2,9/6.9,10/5.5,11/10.4,12/8.3,12/9,12/9.7,13/6.9,13/7.6,22/8.3,22/9,22/9.7,23/10.4,24/7.6,25/5.5,25/6.2,25/6.9,26/9,26/9.7,27/8.3,27/10.4}
  \draw[fill=gray!55] (\s+.08,\y) rectangle (\s+3.92,\y+.45);
\node[font=\scriptsize] at (10.5,5.05) {14 keys};
\node[font=\scriptsize] at (25,5.05) {12 keys};
\node[left] at (-.2,3) {$X(t)$};
\draw[densely dotted,gray] (0,2.5) -- (32,2.5);
\draw[thick] (0,2.5) -- (1,2.5) -- (2,2.5) -- (3,3) -- (4,3.5) -- (5,4) -- (6,4.5) -- (7,5) -- (8,4.5) -- (9,3.5) -- (10,3) -- (11,3) -- (12,3) -- (13,2) -- (14,1.5) -- (15,1.5) -- (16,1.5) -- (17,2) -- (18,2.5) -- (19,3) -- (20,3) -- (21,3.5) -- (22,4) -- (23,3) -- (24,3) -- (25,3) -- (26,2) -- (27,1.5) -- (28,1) -- (29,1) -- (30,1.5) -- (31,2) -- (32,2.5);
\node[font=\scriptsize] at (0,0.15) {$0$};
\node[font=\scriptsize] at (7,0.15) {$7$};
\node[font=\scriptsize] at (14,0.15) {$14$};
\node[font=\scriptsize] at (22,0.15) {$22$};
\node[font=\scriptsize] at (28,0.15) {$28$};
\node[font=\scriptsize] at (32,0.15) {$32$};
\node[font=\scriptsize] at (10.5,-.6) {$X(14)-X(7)=-7$};
\node[font=\scriptsize] at (25,-.6) {$X(28)-X(22)=-6$};
\draw[|-|] (32.8,5) -- node[right,font=\scriptsize] {$D=8$} (32.8,1);
\end{tikzpicture}
\caption{\label{fig:intervals}Two disjoint intervals of slots that force
bumping separately, with $32$ keys, $32$ slots, and reaches of $W=4$ slots
(the slice starts are random, but selected so as to show two large dips of
$X$). Top: the slots; middle: the reaches of the keys, each drawn as a bar from
its first to its last slot (in several lines, for readability), darker for
the keys whose reach starts in one of the intervals $[7\..14)$ and
$[22\..28)$ (shaded); bottom: $X$, with some values of $t$ below. On the two
intervals $X$ decreases by $7$ and $6$: the reaches of $14$ and $12$ keys,
respectively, start in them, and these keys can be placed only in the $10$
and $9$ slots indicated by the braces, so at least $4$ and $3$ of them,
respectively, are bumped. Thus, every placement bumps at least $7$ keys,
which is the minimum (Proposition~\ref{prop:greedy}), whereas
Proposition~\ref{prop:bump} guarantees only $D-W+1=5$ bumped keys, as the
range of $X$ is $D=8$.}
\end{figure}
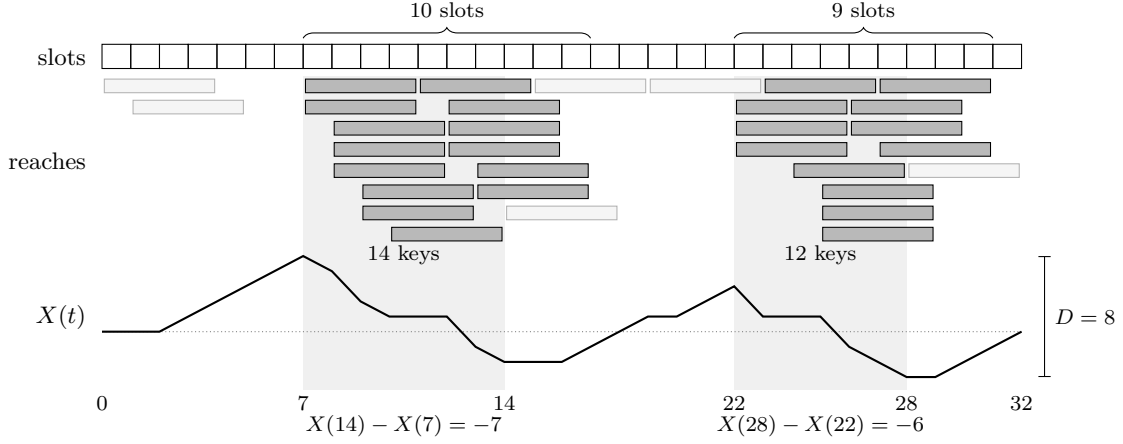

Proposition~\ref{prop:bump} uses a single interval of slots, the one between
the points at which $X$ attains its minimum and its maximum, and thus it
accounts only for the largest imbalance between slots and keys: it cannot
exceed the range of $X$, which grows only as $\sqrt n$, and it is a good
estimate only as long as $\sqrt{\pi n/2}$ is not much larger than $W$, that
is, up to a few million keys when $W=1024$. However, the argument of its
proof applies to every interval of slots $[t\..t')$, and to every placement:
if $X$ increases by $\Delta$ on $[t\..t')$, at least $\Delta-W+1$ slots in
$[t\..t')$ are free, and if $X$ decreases by $\Delta$, at least
$\Delta-W+1$ keys whose reach starts in $[t\..t')$ are bumped. When $X$
oscillates, several intervals have a large imbalance, and each forces its own
free slots or bumped keys (Figure~\ref{fig:intervals}). If the intervals of
slots $[t_0\..t'_0)$, $[t_1\..t'_1)$, $\ldots$ are pairwise disjoint, a free
slot lies in at most one of them, and the reach of a bumped key starts in at
most one of them, so the free slots and the bumped keys counted for
different intervals are distinct: the bumped keys are at least the sum of
the bounds given by the intervals on which $X$ decreases, and the free slots,
which are as many, at least the sum of the bounds given by the intervals on
which $X$ increases. Thus, at least
\[
\frac12\sum_i\Bigl(\bigl|X(t'_i)-X(t_i)\bigr|-W+1\Bigr)^+
\]
keys are bumped, where $x^+=\max(x,0)$. The increment of $X$ on an interval
of $d$ slots is asymptotically normal with mean zero and variance $d$,
so applying the bound to the $n/d$ consecutive intervals of $d$ slots
shows that, for large $W$, the expected number of bumped keys is at least
about $\bigl(n/2d\bigr)\mathbf E\bigl[(\sqrt d\,|Z|-W)^+\bigr]$, where $Z$
is a standard normal variable. Letting $d=W^2/a^2$, this is
$a\bigl(\varphi(a)-a(1-\Phi(a))\bigr)n/W$, where $\varphi$ and $\Phi$ are
the density and the distribution function of $Z$; its maximum, attained for
$a\approx0.61$, that is, for intervals of about $2.7W^2$ slots, is about
\[
0.1\,\frac nW,
\]
a constant fraction of the keys, however large the key set.

The minimum number of bumped keys can actually be computed exactly:

\begin{proposition}
\label{prop:greedy}
Under the assumptions of Proposition~\ref{prop:bump}, the minimum number of
keys bumped by an injective placement is the number of keys bumped by the
following greedy placement: for each slot, in increasing order, place in the
slot the unplaced key with the smallest $\sigma_x$ among those whose
reach contains the slot, if any (breaking ties arbitrarily).
\end{proposition}

\begin{proof}
Let $G$ be the greedy placement and $P$ any placement. We show by induction
on $t$ that there is a placement $Q$ that places at least as many keys as $P$
and agrees with $G$ on the slots in $[0\..t)$, that is, places the same keys
in them. For $t=0$ we take $Q=P$; for $t=n$, $Q$ places the same keys as
$G$, which thus places at least as many keys as $P$.

Let $Q$ agree with $G$ on $[0\..t)$. If $G$ leaves slot $t$ free, every key
whose reach contains $t$ is placed by $G$, and thus by $Q$, before $t$, so
$Q$ leaves slot $t$ free, too. Otherwise, let $g$ be the key placed by $G$ in
slot $t$, which is not placed by $Q$ before $t$. If $Q$ places $g$ in slot
$t$, there is nothing to do. If $Q$ leaves slot $t$ free, we move $g$ to
slot $t$ (or add it, if $Q$ bumps it). Otherwise, $Q$ places in slot $t$ a
key $y\neq g$, which is not placed by $G$ before $t$, so $\sigma_g\leq\sigma_y$
by the greedy choice. If $Q$ bumps $g$, we replace $y$ with $g$; if $Q$
places $g$ in a slot $t'>t$, we exchange $g$ and $y$, as $y$ can be placed
in slot $t'$: $\sigma_y\leq t<t'<\sigma_g+W\leq\sigma_y+W$. In all cases the
number of placed keys does not decrease, the slots in $[0\..t)$ are
unchanged, and the new placement agrees with $G$ on $[0\..t]$.
\end{proof}

The greedy placement is a special case of Glover's algorithm for maximum
matchings in convex bipartite graphs~\cite{GloMMCBG}, which assigns each
slot to the unassigned key whose interval ends first (as all reaches have
the same length, the reach that ends first is the one that starts first),
so Proposition~\ref{prop:greedy} is a special case of Glover's theorem; we
have included its short proof for completeness.

Concretely, the greedy placement is a queue that at each slot receives the
keys whose slice starts there, places the oldest waiting key, and bumps the
keys that have waited $W$ slots. Since keys are placed in order of slice
start and all reaches have the same length, a key is placed if and only if
fewer than $W$ keys that precede it in this order are still waiting when it
arrives: we can thus bump keys as they arrive, and describe the queue by the
number $Q_t\in[0\..W)$ of keys waiting at the beginning of slot $t$. If
$A_t$ keys arrive at slot $t$, then $K_t=(Q_t+A_t-W)^+$ of them are bumped,
and
\[
Q_{t+1}=Q_t+A_t-K_t-1+I_t,
\]
where $I_t$ is one if $Q_t+A_t=0$ (i.e., if slot $t$ remains free) and zero
otherwise. On large key sets the number of keys whose slice starts at a given
slot is approximately a Poisson variable with mean one, and the numbers for
different slots are approximately independent, so we model them with
independent Poisson variables:

\begin{proposition}
\label{prop:poisson}
If the $A_t$ are independent Poisson variables with mean one, the expected
number $\mu$ of keys bumped at a slot in the stationary regime of the queue
satisfies
\[
\frac1{2W+e-2}\leq\mu\leq\frac1{2W}.
\]
\end{proposition}

\begin{proof}
The queue is a finite, irreducible and aperiodic Markov chain on $[0\..W)$,
so it has a unique stationary distribution. Let $Q$ have this distribution,
let $A$ be a Poisson variable with mean one independent of $Q$, let
$K=(Q+A-W)^+$, and let $I$ be one if $Q+A=0$ and zero otherwise, so that
$Q+A-K-1+I$ has the same distribution as $Q$. Taking expectations, $\mathbf E[I]=\mathbf E[K]=\mu$.
Since $KI=0$, $(Q+A-1)I=-I$, and $(Q+A-1)K=(W-1)K+K^2$ (as $Q+A=W+K$ when
$K>0$), squaring gives
\[
(Q+A-K-1+I)^2=(Q+A-1)^2-K^2-2(W-1)K-I,
\]
and since $\mathbf E\bigl[(Q+A-1)^2\bigr]=\mathbf E[Q^2]+\operatorname{Var}[A]=\mathbf E[Q^2]+1$, taking
expectations we obtain
\[
(2W-1)\,\mu+\mathbf E[K^2]=1.
\]
As $K$ is a nonnegative integer, $K^2\geq K$, which yields the upper bound.
For the lower bound, we show that $\mathbf E[K^2]\leq(e-1)\,\mu$. Given
$Q=q$, we have $K=(A-m)^+$ with $m=W-q\geq1$; let $p_j=\Pr(A=m+j)$. Since
$p_{j+1}/p_j=1/(m+j+1)\leq1/(j+2)$, we have $p_j\leq2p_1/(j+1)!$ for
$j\geq1$, with equality when $m=1$. As $j(j-e+1)$ is negative only for $j=1$,
\[
\sum_{j\geq1}j(j-e+1)\,p_j\leq p_1\sum_{j\geq1}j(j-e+1)\frac2{(j+1)!}=0,
\]
where the last sum is zero because when $m=1$ we have $\sum_j j\,p_j=\mathbf
E[(A-1)^+]=1/e$ and $\sum_j j^2p_j=\mathbf E\bigl[\bigl((A-1)^+\bigr)^2\bigr]=1-1/e$.
Thus, $\mathbf E[K^2\mid Q=q]\leq(e-1)\,\mathbf E[K\mid Q=q]$, and averaging
over $q$ and substituting into the identity above completes the proof.
\end{proof}

The proof makes precise the intuition that in the stationary regime the
variance added to the length of the queue by the arrivals must be balanced
by the losses at $W$, each of which reduces its second moment by about $2W$.
By the ergodic theorem, in this model the fraction of bumped keys on large
key sets tends to $\mu$, which is essentially $1/2W$: the minimum number of
bumped keys is thus about
\[
\frac n{2W},
\]
that is, $0.049\%$ of the keys when $W=1024$. Table~\ref{tab:floor} shows
that the model is accurate on the slice starts of our key sets: on $10^9$
keys the minimum computed by the greedy placement is $0.050\%$ of the keys,
whereas Proposition~\ref{prop:bump} guarantees $0.003\%$.

In terms of space, this floor is negligible: by~(\ref{eq:space}), bumping a
fraction $\beta\approx1/2W$ of the keys costs about
$\beta(2+\lg(1/\beta))\approx0.006$ bits per key when $W=1024$. Its
importance lies in the fact that it applies to \emph{every} placement:
beyond a few million keys, no amount of search can make a first level with
$m=n$ free of bumping. Consistently, the variant of \PHast without bumping of
the reference implementation (\texttt{NBFunction}) can be built in practice
only with an output range larger than the number of keys, that is, when the
function is not minimal: with $(1+\varepsilon)n$ slots the queue above has a
negative drift, and its losses decrease exponentially with
$\varepsilon W$.\footnote{Without the constraint of the reaches, space lower
bounds for non-minimal perfect hash functions are classical~\cite{FKSSSFPHF},
and tight lower bounds for non-minimal $k$-perfect hash functions have been
given by Groot Koerkamp et al.~\cite{GHSWNMKPHTLBAFSHT}.}

On the other hand, \PHastR bumps about ten times as many keys as the reaches
force: most of its bumping is due to the constraints on the seeds, which
offer at most $2^S-1$ placements of the keys of a bucket, rather than to the
reaches. By~(\ref{eq:space}),
a search reaching the floor would reduce the space of \PHastR with $S=8$ and
$\lambda=4.25$ from $1.94$ to about $1.89$ bits per key: this is the room
that is left for better searches.

The floor also rules out a version of \PHastR without bumping, whose space
would be just the $S/\lambda$ bits per key of the seeds, built by a
backtracking search in the style of CONSENSUS~\cite{LSWCSESAMPH}, which
combines the search for the seeds with their encoding and gets arbitrarily
close to the lower bound for MPHFs: on a large key set every placement bumps
keys, so no search, however clever, can avoid bumping. Such a search could
succeed only if the keys were first split into subsets of exactly known size
$C$, each placed cyclically in its own $C$ slots, small enough that the
reaches do not force bumping: by Proposition~\ref{prop:bump}, as the expected
range of $X$ on a subset is about $\sqrt{\pi C/2}$, this happens only if
$\sqrt{\pi C/2}$ is significantly smaller than $W$. The slots must be used
cyclically, because otherwise those at the ends of a subset can be reached
by few keys: if the slices of the keys start uniformly in $[0\..C-W+1)$, the
slots at each end receive enough keys with probability about $2W/C$, so a
placement without bumping exists with probability at most about
$(2W/C)^2$. Such a split is the first step (albeit for
different reasons) of Bucketed CONSENSUS-RecSplit~\cite{LSWCSESAMPH}, which
uses a minimal $k$-perfect hash function.

Even on such subsets, however, a depth-first search detects its mistakes too
late. Without bumping every slot must be used, so the natural failure test is
that a free slot can no longer receive a key; but this becomes certain only
after all the buckets whose slices can contain the slot, about $W/\lambda$ of
them, have been placed. The choice that caused a failure might thus have been
made about $W/\lambda$ buckets earlier, and a depth-first search with this
test must exhaust the alternatives of all the buckets placed after it before
it can change it. In a prototype that places about $10^5$ keys cyclically on
as many slots and chains seeds as CONSENSUS does, with
$S=8$, $L=512$, $R=4$, and $\lambda=4$ (so $W/\lambda=128$), the search never
got past the first $160$ buckets, even though seeds are more than two bits
longer than the information-theoretic minimum for a bucket, and with a
budget of $2000$ seeds tried per bucket, counting those tried again after
backtracking.

\section{Conclusions}

In \PHastR the lowest bits of a seed choose one of a few layouts, which are
independent of one another,
like the seeds of the pseudorandom scheme of \PHast, and whose offsets come
for free with the computation of the bucket; the remaining bits choose a
rotation, which moves all the keys of the bucket by the same amount, like
the seeds of the additive scheme of \PHastP. Rotations
keep the search bit-parallel, as the free rotations of a layout are found by
cyclic shifts of the bits of the rings of the keys, which are consecutive bits
once the set of used slots is stored by residue classes. Layouts, in turn,
neutralize the self-collisions of the additive scheme, which, even with
wrapping, cause a large part of the bumping of \PHastP, as a bucket is
affected only if it has a self-collision in every layout; and, once the seed
width and the expected bucket size are fixed, bumping is where both the
excess space and the slow queries come from: queries for the keys of the
first level take about the same time in \PHastR and in \PHastP with
wrapping, and almost all the difference is due to bumped keys
(Table~\ref{tab:split}).

As a result, with essentially the same space as the best
variant of \PHastP with wrapping, \PHastR is built in $37\%$ of the time, and
its queries are $34\%$ faster on $10^9$ keys
(Tables~\ref{tab:1e7}--\ref{tab:1e9}); they are faster also than those of
PtrHash, which uses more than one bit per key more; with a larger expected
bucket size,
it comes within $0.02$ bits per key of the space of \PHast, which takes more
than ten times as long to build, and its queries are still about as fast;
with many threads it is built more than two and a half times faster than \PHastP
(Table~\ref{tab:threads}); and almost every configuration of \PHastP in
Figure~\ref{fig:pareto} is dominated by one of \PHastR. Finally, an offline
construction builds the same structure keeping
the hashes of the keys on disk, with less than half a byte per key of memory
on large key sets.

There are, however, intrinsic limits to how few keys can be
bumped: when each key can be placed only within a range of $W$ slots, every
placement of random keys must bump, on large key sets, a constant fraction
of the keys (in expectation, at least about $0.1/W$ by our asymptotic lower
bound, and about $1/2W$ for the best placement in a Poisson model), although this floor is about ten times lower than the bumping of \PHastR.

There are several open problems and further research directions. On the
practical side, the small final levels could be replaced by a more compact
structure. More interestingly,
\PHastR bumps about ten times as many keys as the reaches force: it is open
whether a search can approach this floor efficiently (which would save about
$0.05$ bits per key), and whether bump-free placements, which would use just
$S/\lambda$ bits per key, can be found efficiently within subsets of exactly
known size, in spite of the long-range constraints that made simple
backtracking fail in our prototype.

\bibliography{biblio}

\end{document}